\documentclass[11pt]{article}
\usepackage[affil-it]{authblk}
\usepackage{amsthm}
\usepackage{tikz}
\usetikzlibrary{automata,positioning,matrix}
\usepackage{tikz-qtree}
\usetikzlibrary{arrows}

\usetikzlibrary{cd}

\usepackage{epic,eepic,amsmath,amssymb,stmaryrd}

\usepackage{caption}
\usepackage[T1]{fontenc}
\usepackage[utf8]{inputenc}

\theoremstyle{plain}
\newtheorem{theorem}{Theorem}[section]
\newtheorem{lemma}[theorem]{Lemma}

\newtheorem{proposition}[theorem]{Proposition}
\newtheorem{corollary}[theorem]{Corollary}

\theoremstyle{definition}
\newtheorem{definition}[theorem]{Definition}
\newtheorem{remark}[theorem]{Remark}
\newtheorem{example}[theorem]{Example}

\newcommand{\da}{\mathord{\downarrow}}
\newcommand{\rom}[1]{\rm{\uppercase\expandafter{\romannumeral #1}}}
\newcommand{\id}{\mathrm{id}}

\newcommand{\oneset}[1]{\{#1\}}

\newcommand{\real}{\overline{\mathbb R}_+}

\newcommand{\V}{{\mathcal V}_w}

\newcommand{\cc}{\mathrm c}
\newcommand{\s}{\mathrm{s}}

\newcommand{\us}{ {\delta^{\s}}}
\newcommand{\ms}{ {\beta^{\s}}}

\newcommand{\uk}{ {\delta^{k}}}
\newcommand{\mk}{{\beta^{k}}}

\newcommand{\mm}{m}

\newcommand{\vk}{{\mathcal V}_\mathcal{K}}
\newcommand{\vs}{{\mathcal V}_\mathrm{s}}
\newcommand{\D}{\mathcal{D}}
\newcommand{\W}{\mathcal{W}}
\newcommand{\K}{\text{{\fontfamily{cmtt}\selectfont K}}}
\newcommand{\kk}{\mathcal{K}}
\newcommand{\kc}{\mathcal{K}_\mathrm{c}}
\newcommand{\R}{\mathbb{R}_+}
\newcommand{\RE}{\overline{\mathbb{R}}_+}
\newcommand{\clk}{\mathrm{cl}_{\kk}}

\title{ Completing Simple Valuations in $\K$-categories}
\author{Xiaodong Jia$^{a}$$^{b}$\thanks{Xiaodong Jia: jia.xiaodong@yahoo.com}~~and~~Michael Mislove$^{b}$\thanks{Michael Mislove: mislove@tulane.edu }}
\affil{$^{a}$School of Mathematics, Hunan University\\ Changsha, Hunan 401182, China\\
$^{b}$Department of Computer Science, Tulane University\\ New Orleans, LA 70118, USA}

\begin{document}

\maketitle 

\begin{abstract}
We prove that Keimel and Lawson's $\kk$-completion $\kk_\cc$ of the simple valuation monad $\vs$  defines a monad $\kk_\cc\circ \vs$ on each $\K$-category $\kk$.   We also characterise the Eilenberg-Moore algebras of $\kk_\cc\circ\vs$  as the weakly locally convex $\kk$-cones, and its algebra morphisms as the continuous linear maps. 

In addition, we explicitly describe the distributive law of $\vs$ over $\kk_\cc$, which allows us to show that the $\kk$-completion of any locally convex (resp., weakly locally convex, locally linear) topological cone is a locally convex (resp., weakly locally convex, locally linear) $\kk$-cone. 

We also give  an example -- the Cantor tree with a top -- that shows the dcpo-completion  of the simple valuations is not the $\D$-completion of the simple valuations in general, where $\D$ is the category of monotone convergence spaces and continuous maps. 

\emph{Keywords}:  Valuation powerdomain, simple valuations, $\kk$-completion, monads, Eilenberg-Moore algebras, topological cones
\end{abstract}

\section{Introduction}

In domain theory, the probabilistic powerdomain monad on the category of continuous domains and Scott-continuous maps, proposed by Jones and Plotkin \cite{jones89, jones90}, is the most widely used  mathematical construct for denoting probabilistic computation in denotational semantics. It was realised by Kirch~\cite{kirch93}, Tix~\cite{tix95} and many others that the weak topology on the probabilistic powerdomain is more natural than the Scott topology, and for this reason, the extended probabilistic powerdomain was proposed over the category of $T_{0}$ spaces and continuous maps~\cite{cohen06}. Indeed, the weak topology behaves more transparently than the Scott topology on general dcpos as well as topological spaces. For example, Alvarez-Manilla, Jung and Keimel~\cite{alvarez04}  proved that the probabilistic powerdomain of a stably compact space is again stably compact in the weak topology. Schr\"oder and Simpson~\cite{schroder05} showed that continuous linear functionals on the extended probabilistic powerdomain are uniquely determined by continuous functions from the underlying space to the reals. Detailed proofs of these results can be found in~\cite{goubault15, keimel12}.

The extended probabilistic powerdomain over a topological space $X$ consists of the so-called continuous valuations on $X$ (see Definition~\ref{defofvaluations}), and among all continuous valuations, the simple valuations are the most natural ones. They are of the form $\Sigma_{i= 1}^{n}r_{i}\delta_{x_{i}}$, where $r_{i}, i=1, .. ,n$ are real numbers and $\delta_{x}$ is the Dirac measure at $x$ for $x\in X$. It was proved by Jones~\cite{jones90} that for a continuous domain $P$ with the Scott topology, every continuous valuation on $P$ can be written as a directed supremum of simple valuations. This is not true for general dcpos, but simple valuations are usually regarded as ``bricks'' that reflect  properties of  the entire ``building'' of continuous valuations. For this reason, it is useful to single out classes of spaces (that include more general dcpos than the continuous domains) for which there is a completion of the simple valuations $\vs$ that also forms a monad. The purpose of this paper is to show that each $\K$-category admits such a monad. 

Heckmann~\cite{heckmann96} considered the simple valuations in the topological setting, where he proved the sobrification of the space of simple valuations on a topological space~$X$ consists of point-continu\-ous valuations on~$X$. Again, he used the weak topology rather than the Scott topology, and he proved that such a sobrification gives rise to the free weakly locally convex sober topological cone over~$X$. Goubault-Larrecq and the first author~\cite{goubault19} proved that the sobrification of the space of simple valuations also defines a monad over topological spaces. Keimel and Lawson~\cite{keimel09a} realised that sobrification of a space is just a special case of $\kk$-completion of that space, which led the authors to consider whether there is a general distributive law between $\kk$-completion and the simple valuation powerdomain. As is well known, such a distributive law can be used to show that their composition  is again a monad. We show that this distributive law does exist, so the $\kk$-completion of the space of simple valuations gives a monad on the category of $\kk$-spaces for each $\K$-category~$\kk$.  Generalising Heckmann's result, we prove that the $\kk$-completion of the simple valuations over a space $X$ gives rise to the free weakly locally convex topological $\kk$-cone over $X$. This result would enable us to prove that the $\kk$-completion of any locally convex  (resp., weakly locally convex, locally linear)  topological cone is a locally convex  (resp., weakly locally convex, locally linear) $\kk$-cone. 

The paper is organised as follows: we recall the extended powerdomain construction, topological cones and $\kk$-completions in Section~2, where we also prove many properties of general $\kk$-completions that are useful in our later discussion. In Section~3, we give two approaches to proving that the $\kk$-completion of the simple valuation monad gives a monad over topological spaces. In Section~4, we characterise the Eilenberg-Moore algebras of our new monad as the weakly locally convex $\kk$-cones, and employ it in Section~5 to prove that the $\kk$-completion of a locally convex (resp., weakly locally convex, locally linear) topological cone is a locally convex (resp., weakly locally convex, locally linear) $\kk$-cone. In the final section, we  show that ${\mathcal C}^\top$, the Cantor tree with a top element, has the property that the dcpo-completion of its poset of simple valuations is not the monotone convergence completion of the space~$\vs\, {\mathcal C}^\top$ endowed with the weak topology.

\section{Preliminaries}

We use standard notions and terminology from Domain Theory~\cite{gierz03,abramsky94} and from non-Hausdorff topology~\cite{goubault13a}. 

In this paper, we restrict attention to the category ${\bf TOP_{0}}$ of $T_0$ spaces and continuous maps, so all topological spaces are assumed to be $T_{0}$ spaces unless stated otherwise.  For a poset $P$, we use $\Sigma P$ to denote the space $P$ equipped with the Scott topology; and for a topological space $X$, we use $\Omega X$ to denote the poset $X$ equipped with the specialisation order on $X$, that is, $x\leq y$ for $x, y\in X$ if and only if $x$ is in the closure of $\{y\}$. We use $\mathbb R_+$ to denote the set of nonnegative reals, and $\real$ the set of nonnegative  reals extended with $\infty$. Whenever  $\real$ is treated as a topological space, we mean that it is equipped with the Scott topology. Note that this implies a function $f\colon X\to \real$ is continuous iff it is lower semicontinuous. For a topological space $X$, we use $X^{\s}$ to denote the canonical sobrification of $X$, which we equate to the set of irreducible closed subsets of $X$  endowed with the lower Vietoris topology (see for example, \cite[Definition 8.2.17]{goubault13a}). We use $\eta^{\s}_{X}$ to denote the topological embedding of $X$ into $X^{\s}$ which sends each $x\in X$ to the closure of $\{x \}$ which is $\da x$ in the specialisation order. Notice that $\eta^{\s}_{X}$ is an embedding if and only if $X$ is $T_0$.  We use $\mathbb S$ to denote the Sierpi\'nski space consisting of two elements $0$ and $1$ in which the only proper open set is the singleton $\{ 1\}$. Finally, for a subset $A$ of a space $X$ we use $\chi_{A}$ to denote the characteristic function of~$A$, that is, $\chi_{A}(x)=1$ when $x\in A$ and $0$ otherwise.

\subsection{The extended probability powerdomain monad $\V$}

\begin{definition}
\label{defofvaluations}
A \emph{valuation} on a topological space $(X,\mathcal OX)$ is a function~$\mu$ from $\mathcal OX$ to the extended nonnegative reals $\RE$ satisfying for any $U, V\in \mathcal OX$:
\begin{itemize}
\item (strictness)  $\mu (\emptyset) = 0$;   
\item (monotonicity) $\mu(U)\leq \mu(V)$ if $U\subseteq V$; 
\item (modularity) $ \mu(U) + \mu(V) = \mu (U\cup V) +\mu (U\cap V)$.
\end{itemize}

A \emph{continuous valuation} $\mu$ on $(X, \mathcal OX)$ is a valuation that is Scott-continuous from $\mathcal OX$ to $\RE$, that is, for every directed family of open subsets $U_i, i\in I$, it is true that:
\begin{itemize}
\item (Scott-continuity)  $\mu (\bigcup_{i\in I}U_i) = \sup_{i\in i} \mu(U_i)$.
\end{itemize}

Valuations on a topological space~$X$ can be ordered: given $\mu$ and $\nu$, $\mu\leq \nu$ if and only if $\mu(U)\leq \nu(U)$ for all $U\in \mathcal OX$, an order that sometimes is referred to as the \emph {stochastic order}. 

The set of continuous valuations on $X$ endowed with the stochastic order is denoted by~$\mathcal VX$.
\end{definition}

Among all continuous valuations on $X$ there are \emph{Dirac masses} $\delta_x, x\in X$. For each $x\in X$, $\delta_x$ is defined as:
$$\delta_x(U) = \begin{cases}1 & x\in U \\ 0 & \text{otherwise.}  \end{cases}$$
\emph{Simple valuations} are finite linear combinations of Dirac masses. They are of the form $\Sigma_{i=1}^n r_i\delta_{x_i}$, where $x_i,  i=1,...,n$ are in $X$ and $r_i, i=1,..., n$ are nonnegative real numbers, i.e., $r_i\in \mathbb R_+$. For each open set $U$, $(\Sigma_{i=1}^n r_i\delta_{x_i})(U) = \Sigma_{i/ x_i\in U} r_i $.  The set of all simple valuations on $X$  is denoted by~$\vs X$. In contrast of simple valuations, one can consider \emph{finite valuations}, i.e., these continuous valuations that only take finitely many (finite) values. Simple valuations are finite valuations, and it is proved by R. Tix \cite{tix99} that finite valuations on sober spaces are simple valuations. In general, finite valuations are not necessarily simple, but we do have the following:

\begin{proposition}
Let $X$ be a $T_0$ topological space. Then the finite valuations on $X$ are precisely the simple valuations if and only if $X$ is sober.
\end{proposition}
\begin{proof}
The ``if'' (hard) direction is given as \cite[Satz 2.2]{tix95}. For the ``only if'' direction, we assume that $A$ is an irreducible closed subset of $X$ and define a finite valuation $\mu$ as 
$$ \mu (U) = \begin{cases} 1 & U\cap A \not= \emptyset\\ 0 & \text{otherwise}. \end{cases}$$
It is easy to see that $\mu$ is a continuous finite valuation, hence by assumption there exists a finite subset $F=\oneset{x_1, ..., x_n}$ of $X$ and positive reals $r_i, i=1, ..., n$ such that $\Sigma_{i=1}^n r_i\delta_{x_i} = \mu $. It follows that  $\Sigma_{i=1}^n r_i = 1$, and hence for any open set $U$, $U$ intersects $A$ if and only if $F\subseteq U$. Since $X$ is $T_0$, this implies that $x_1 = x_2 =...= x_n$, so $A$ is the closure of $\{x_1\}$ and $x_1$ is the unique point whose closure is $A$.
\end{proof}

For any topological space $X$, we topologise $\mathcal VX$ with the \emph{weak topology}, which is generated by a subbasis of sets of the form 
$$[U>r]= \{\mu\mid \mu \text{ is continuous and }  \mu(U) > r\},$$ 
for $U\in \mathcal OX, r\in \mathbb R_+$. We use $\V X$ to denote the space $\mathcal VX$ equipped with the weak topology and call $\V X$ the \emph{extended probabilistic powerdomain} or the \emph{valuation powerdomain} over $X$. The weak topology on $\V X$ is natural in many ways. For example, the specialisation order of the weak topology is just the stochastic order on $\mathcal VX$ and the canonical map  $\delta_X\colon X \to \V X$, sending each $x\in X$ to the Dirac mass $\delta_x$ at $x$, is a topological embedding.

We can extend $\V$ to an endofunctor on ${\bf TOP_0}$ by defining its action $\V f$ on continuous maps $f\colon X\to Y$ by $\V f(\mu)(V) = \mu(f^{-1}(V))$. Moreover, $\V$ is a monad over the category ${\bf TOP_0}$. In order to show this, we need an integration theory of lower semicontinuous functions with respect to continuous valuations and  Manes' equivalent description of monads. 

For any topological space $X$, every lower semicontinuous function $h\colon X\to \real$ has a Choquet-type \emph{integral} with respect to a continuous valuation $\mu$ on $X$ defined by:
$$ \int_{x\in X} h(x)d\mu = \int_0^\infty \mu(h^{-1}(r,\infty])dr,$$
where the right side of the equation is a Riemann integral. If there is no chance of confusion, we write $\int_{x\in X} h(x)d\mu$ as $\int h~d\mu$. This integral enjoys many nice properties. For example, it is linear: for any continuous valuation~$\mu$ on~$X$, $a, b\in \mathbb R_{+}$, and lower semicontinuous maps $f, g \colon X\to \real$, 
$$\int af + bg~d\mu  = a\int f~d\mu+ b\int g~d\mu. $$
For simple valuations $\Sigma_{i=1}^n r_i\delta_{x_i}$, we also have: 
$$\int f d(\Sigma_{i=1}^n r_i\delta_{x_i})  = \Sigma_{i=1}^{n}r_{i}f(x_{i}),$$ and the function 
$$\mu\mapsto \int fd\mu \colon \V X \to \real $$ is continuous. The reader is referred to \cite{kirch93, lawson04, tix95, goubault19} for more related properties of integration.

\begin{definition}[Manes' description for monads]{\rm \cite{manes76}}
A \emph{monad} on a category {\bf C} is a triple $(T , \eta, \_^\dagger)$ consisting of a map $T$ from objects $X$ of~{\bf C} to objects $T X$ of {\bf C}, a collection $\eta = (\eta_X)_X$ of morphisms $\eta_X : X \to TX$, one for each object $X$ of {\bf C} (called the unit of $T$), and a so-called extension operation $\_^\dagger$ that maps every morphism $ f : X\to TY$ to $f ^\dagger : T X\to T Y$ such that:
\begin{enumerate}
\item $\eta_X^\dagger = \id_{TX}$;
\item for every morphism $f : X\to TY$, $ f^\dagger\circ  \eta_X = f$;
\item  for all morphisms $f : X\to TY$ and $g: Y\to TZ$, $g^\dagger\circ  f^\dagger = (g^\dagger\circ f)^\dag$.
\end{enumerate}
\end{definition}

With the aforementioned ingredients, we can show that $\V$ is a monad over the category ${\bf TOP_0}$. The unit of $\V$ is given by $\delta_X\colon x\mapsto \delta_x$ for each $X$, and for continuous functions $f\colon X\to \V Y$ the extension operation $\_^{\dagger}$ is given by 
$$f^\dagger(\mu)(U) = \int_{x\in X} f(x)(U)d\mu.$$
The function $f^\dagger\colon \V X\to \V Y: \mu\mapsto (U\mapsto \int_{x\in X} f(x)(U)d\mu)$, in particular, is continuous. Alternatively, one can describe the multiplication $\beta\colon \V^2\to\V$ of the monad at $X$ as the map $\beta_{X}$ sending each continuous valuation $\varpi\in \V(\V X )$ to $\id_{\V X}^\dagger(\varpi)=(U\mapsto \int_{\mu\in \V X} \mu(U)d\varpi)$. For a detailed discussion the reader is referred to Section 2.3 in \cite{goubault19}.

Restricting ourselves to the simple valuations, we conclude that $\vs$ is a monad over the category ${\bf TOP_0}$, where for a topological space $X$, $\vs X$ is the subspace of $\V X$ consisting of simple valuations; and for continuous maps $f\colon X\to Y$ and simple valuations $\Sigma_{i=1}^n r_i\delta_{x_i}$, $\vs f (\Sigma_{i=1}^n r_i\delta_{x_i}) = \V f  (\Sigma_{i=1}^n r_i\delta_{x_i})= \Sigma_{i=1}^n r_i\delta_{f(x_i)}$. The unit of $\vs$ at $X$ is the map ${\us}_{X}\colon X\to \vs X\colon x\mapsto \delta_x$, the corestriction of the map $\delta_{X}$ on $\vs X$. The multiplication $\ms\colon \vs^2\to \vs$ at $X$ is the map $\ms_{X}(\Sigma_{i=1}^{n}r_{i}\delta_{\mu_{i}}) = \Sigma_{i=1}^{n} {r_{i} \mu_{i}}$, where for $i= 1,\ldots,n$, $\mu_{i}$ is a simple valuation and $r_{i}$ is a nonnegative real number.

While a characterisation of the Eilenberg-Moore algebras of the $\V$-monad remains an open problem, we know that $\vs$-algebras are precisely the \emph{weakly locally convex topological cones}~\cite{goubault19} : For any weakly locally convex topological cone $X$, the structure map $\alpha_X \colon \vs X\to X$ is the continuous map which sends each simple valuation
$\Sigma_{i=1}^n r_i\delta_{x_i}$ to its \emph{barycentre} $\Sigma_{i=1}^n r_i x_i$. Moreover, the $\vs$-morphisms are continuous \emph{linear maps}. The notions of topological cone and linear map are explained below, and the barycentre of a valuation is defined in Definition~\ref{defn:bary:choquet}.

\subsection{Ordered cones and topological cones}

The following notions are mainly from \cite{keimel08} and \cite{goubault19}.

\begin{definition}[Cone]
A \emph{cone} is defined to be a commutative monoid~$C$ together with a scalar multiplication by nonnegative real numbers satisfying the same axioms as for vector spaces; that is, $C$ is endowed with an addition $(x, y)\mapsto x+y: C\times C\to C$ which is associative, commutative and admits a neutral element~$0$, and with a scalar multiplication $(r, x)\mapsto r \cdot x: 
\mathbb R_+ \times C\to C$ satisfying the following axioms for all $x, y\in C$ and all $r, s\in \mathbb R_+$:
\begin{align*}
&r\cdot (x+y) = r\cdot x+r\cdot y         & &  (rs)\cdot x=r\cdot (s\cdot x)        &  0\cdot x = 0 \\
& (r+s)\cdot x = r\cdot x+s\cdot x       & &1\cdot x=x                                     &r\cdot 0 = 0
\end{align*}
We shall often write $rx$ instead of $r \cdot x$ for $r\in \mathbb R_+$ and $x\in C$.

An \emph{ordered cone} is a cone $C$ endowed with a partial order $\leq$ such that addition and multiplication by fixed scalars $r\in \mathbb R_{+}$ are order preserving, that is, for all $x, y, z\in C$ all $r\in \mathbb R_{+}$: $x\leq y \Rightarrow x+ z\leq y+z $ and $rx\leq ry $. 

A \emph{semitopological cone} is a cone with a $T_0$ topology that makes $+$ and $\cdot$ separately continuous.

A \emph{topological cone} is a cone with a $T_0$ topology that makes $+$ and $\cdot$ jointly continuous.
\end{definition}

Topological cones are semitopological cones, semitopological cones are ordered cones in their specialisation order. The extended reals $\real$ with the Scott topology is a topological cone, with scalar multiplication and addition extended as: $a+\infty = \infty$ for all $a\in \real$, $b\cdot \infty=\infty$ for $b\in \real\setminus \{0\}$ and $0\cdot\infty=0$. Usually a cone cannot be embedded into a vector space as a subcone,  the reason is that a cone might not satisfy the Cancellation Law (a+b = a+c implies b=c), a property enjoyed by every vector space. As an example, the extended reals $\real$ does not satisfy the Cancellation Law hence cannot be  be embedded into any vector space. 

The following definition is taken from~\cite[Definition 3.7]{keimel08}. 

\begin{definition}
A function $f:C\to D$ from cone $C$ to $D$ is called \emph{linear} if it is \emph{homogeneous}:
$$f(ra)= rf(a)~\text{for all}~a\in C~\text{and}~r\in \mathbb R_{+} $$ 
and \emph{additive}:
$$f(a+b)=f(a)+f(b)~\text{for all}~a, b\in C.$$
If $D$ is an ordered cone, $f$ is called \emph{superadditive} if
$$ f(a+b)\geq f(a)+f(b)~\text{for all}~a, b\in C$$
and \emph{subadditive} if
$$ f(a+b)\leq f(a)+f(b)~\text{for all}~a, b\in C.$$
We say that $f$ is \emph{sublinear} (resp., \emph{superlinear}), if f is homogeneous and subadditive (resp., superadditive).
\end{definition}

\begin{definition} Let $C$ be a cone. 
\begin{itemize}
\item A subset $A\subseteq C$ is \emph{convex}, if given $a, b\in A$, the linear combination $ra+(1-r)b\in A$ for any $r\in [0,1 ]$.
\item A subset $A\subseteq C$ is a \emph{half-space}, if both $A$ and its complement are convex. 
\item A cone $C$ with a $T_0$ topology is called \emph{weakly locally convex}\footnote{This was introduced as local convexity in~\cite{heckmann96}.}, if for every point $x\in C$, every open neighbourhood $U$ of $x$ contains a convex (not necessarily open) neighbourhood of $x$.
\item A cone $C$ with a $T_0$ topology is called \emph{locally convex}, if each point has a neighbourhood basis of open convex neighbourhoods.
\item A cone $C$ with a  $T_0$ topology is called \emph{locally linear}, if $C$ has a subbase of open half-spaces.
\end{itemize}
\end{definition}

It is immediate from the definition that every locally linear (semi)topological cone is locally convex and every locally convex (semi)topological cone is weakly locally convex. 

\begin{proposition}
For any topological space $X$, $\V X$ has a canonical cone structure that makes it a locally linear \emph{topological cone} (see, for example, {\rm \cite[Proposition 3.8]{goubault19}}).
\end{proposition}
\begin{proof}
For any $\mu, \nu\in \mathcal VX, r\in \mathbb R_+$ and $U$ open in $X$, define $(\mu+\nu)(U)= \mu(U)+\nu(U)$ and $(r\cdot\mu)(U)=r\cdot\mu(U)$. The \emph{zero valuation} $(U\mapsto 0)$ is the neutral element. Then the triple $(\mathcal VX, +, \cdot)$ is a cone. The addition $ + \colon \V X\times \V X\to \V X\colon (\mu, \nu)\mapsto \mu+\nu $ and scalar multiplication $r\cdot\_\colon \V X\to \V X\colon \mu\mapsto r\cdot\mu$ for each $r\in \mathbb R_+$ are jointly continuous.  The latter fact implies that $\V X$ is a topological cone. Notice that the sub-basic opens $[U>r]$ of $\V X$ are all half spaces, hence $\V X$ is locally linear, hence also locally convex, and also weakly locally convex. 
\end{proof}

Also, note that $\V f\colon \V X\to \V Y$ is continuous and linear,  for any continuous map $f\colon X\to Y$.

Let $(C, +,\cdot)$ be a cone, by a \emph{subcone} $A$ of $C$ we mean that the neutral element~$0$ of $C$ is in $A$, and that for any $a, b\in A, r\in \mathbb R_+$, we have $a+b\in A$ and $r\cdot a\in A$. Following this terminology, the cone $\vs X$ of simple valuations on~$X$ is a subcone of~$\V X$, and $\vs X$ with the subspace topology is also a locally linear topological cone.

We also need the notion of a linear retraction between semitopological cones. In topology, a continuous function $f\colon X\to Y$ is called a \emph{retraction}
if there exists a continuous map $g\colon Y\to X$ such that $f\circ g= \id_Y$. The function~$g$ is called a \emph{section} (of $f$), and $Y$ is called a \emph{retract} of $X$. 
For semitopological cones $C$ and $D$, a retraction map $f\colon X\to Y$ is called a \emph{ linear retraction} if $f$ is linear. Notice that we do not require the section of $f$ to be a linear map. 

\begin{lemma}\label{retractionofweaklylc}{\rm \cite[Proposition 6.6~]{heckmann96}}
Let $C, D$ be topological cones, and $f\colon C\to D$ be a linear retraction. If $C$ is weakly locally convex, then so is $D$. \hfill $\Box$
\end{lemma}

\subsection{The $\kk$-completion}

In order to prove that the category of  monotone convergence spaces is reflexive in ${\bf TOP_0}$, Keimel and Lawson considered the general setting of \emph{$\K$-categories}~\cite{keimel09a}. 

\begin{definition}[$\K$-\emph{category}] \label{defk}
A full subcategory $\mathcal K$ of the category ${\bf TOP_0}$ is called a 
$\K$-\emph{category} if all objects of $\mathcal K$ satisfy the following properties:
\begin{enumerate}
  \item Homeomorphic copies of $\kk$-spaces are $\kk$-spaces. That is, $\kk$ is a replete subcategory of ${\bf TOP_{0}}$.
  \item All sober spaces are $\kk$-spaces. That is, all sober spaces are in the category $\kk$.
  \item In a sober space $S$, the intersection of any family of $\kk$-subspaces is a $\kk$-space. 
  \item Continuous maps $f\colon S\to T$ between sober spaces $S$ and $T$ are \emph{$\kk$-continuous}, that is, for every $\kk$-subspace $K$ of $T$, the inverse image
           $f^{-1}(K)$ is a $\kk$-subspace of $S$ or, equivalently, $f(\clk(Z))\subseteq \clk f(Z)$ for every subset $Z\subseteq S$, where for a subset $A$ of a sober space $X$, $\clk A$ stands for the smallest $\kk$-subspace of $X$ containing $A$. The existence of $\clk A$ is guaranteed by Item 3. 
\end{enumerate}
\end{definition}
Note that a $\K$-category $\kk$ is also referred to as \emph{having {Property}~$\K$}, and  the objects of $\kk$ are called $\kk$-spaces.

Keimel and Lawson \cite{keimel09a} proved that every $\K$-category~$\kk$ is a reflective subcategory of ${\bf TOP_0}$: The corresponding reflector $\kc$ sends each topological space~$X$ to the smallest $\kk$-subspace of~$X^\s$ that contains $\eta^{\s}_{X}(X)$, the embedding image of $X$ into $X^{\s} $. They then showed that the category of monotone convergence spaces, which we denote by $\D$, is a $\K$-category. 

Additionally, the category $\mathcal S$ of sober spaces is a $\K$-category, the related reflector $\mathcal S_\cc$ is just the sobrification functor.  It was proved in \cite{wu20} that the category $\W$ of well-filtered spaces is also a $\K$-category. However, the category of Scott spaces (dcpos equipped with the Scott topology) and Scott-continuous functions  is not a $\K$-category since not every sober space is a Scott space: the topology on a sober space need not be the Scott topology of its specialisation order; the unit interval $[0, 1]$ with the ordinary topology is Hausdorff hence sober, but it is not a Scott space. 
\begin{remark}
In this paper, $\kk$ is always reserved for some $\K$-category. Moreover, $\kc$ is used to denote the corresponding reflector.
\end{remark}

\begin{figure}[h]
\centering
\begin{tikzpicture}[commutative diagrams/every diagram]
\matrix[matrix of math nodes, name=m, commutative diagrams/every cell] {
\kc(X)  &  K \\
 X &    \\};
  \path[commutative diagrams/.cd, every arrow, every label]
    (m-1-1) edge["$\overline f$"] (m-1-2)
    (m-2-1) edge[swap,"$f$"] (m-1-2)
                 edge[commutative diagrams/hook, "$\eta_X$"]  (m-1-1);
\end{tikzpicture}
\begin{tikzpicture}[commutative diagrams/every diagram]
\matrix[matrix of math nodes, name=m, commutative diagrams/every cell] {
 \kc(X)  &  \kc(Y) \\
 X &  Y   \\};
  \path[commutative diagrams/.cd, every arrow, every label]
    (m-1-1) edge["$\kc (f)$"] (m-1-2)
    (m-2-1) edge[swap,"$f$"] (m-2-2)
                 edge[commutative diagrams/hook, "$\eta_X$"]  (m-1-1)
    (m-2-2) edge[commutative diagrams/hook, "$\eta_Y$"]  (m-1-2);
\end{tikzpicture}
\end{figure}

Since each $\K$-category $\kk$ is a reflective subcategory of $\bf TOP_0$, the reflector $\kk_\cc\colon {\bf TOP_0}\to \kk$ is left adjoint to the inclusion of $\kk$ into $\bf TOP_0$. We now state some properties that apply in this situation that will be useful in what follows.

For a space $X$, we use $\eta_{X}\colon X\to \kc(X) $ to denote the corestriction of $\eta^{\s}_{X}$ on $\kc(X)$ throughout the paper. Then $\eta_{X}$ is the unit for the monad~$\kk_\cc$, and it is a topological embedding of $X$ into $\kc(X)$. Moreover, the pair $(\kc(X), \eta_{X})$ satisfies the following universal property:  for each $\kk$-space $K$ and continuous map $f\colon X\to K$, there is a unique extension map $\overline f \colon \kc (X)\to K$ of $f$ such that $\overline f\circ \eta_{X} = f$. We call $(\kc(X), \eta_{X})$ is the \emph{canonical $\kk$-completion} of $X$. More generally, we say a pair $(Y, g)$ is a \emph{$\kk$-completion} of $X$ if $Y$ is a $\kk$-space, $g\colon X\to Y$ is a topological embedding, and for every continuous map $f$ from $X$ to a $\kk$-space $K$, there exists a unique extension map $\overline f\colon Y\to K$ such that $\overline f\circ g=f $. 

As a functor from $\bf TOP_0$ to $\kk$, the  reflector $\kc$ sends each  space $X$ to $\kc(X)$, and each continuous map $f\colon X\to Y$ to $\kc(f)=\overline{\eta_{Y}\circ f}$, the unique continuous function from $\kc(X)$ to $\kc(Y)$ that extends $\eta_{Y}\circ f$, i.e., $\kc(f) \circ\eta_X=\eta_Y \circ f$.

\begin{remark}
Restricting to the category $\D$ of monotone-convergence spaces, Keimel and Lawson showed that $\D$-completion gives rise to, in a natural way, the dcpo-completion, a concept well studied in~\cite{zhao10}. In detail, for a poset~$P$, the pair $(\overline P, e)$ is called a \emph{dcpo-completion} of~$P$, if $e\colon P\to \overline P$ is a topological embedding with respect to the Scott topologies on both $P$ and $\overline P$, where $\overline P$ is a dcpo,  and for every Scott-continuous map $f$ from $P$ to a dcpo $L$, there exists a unique Scott-continuous map $\overline f$ such that $\overline f\circ e = f$. They proved that the functor $\Omega\circ \D_{\cc}\circ \Sigma$ is the left adjoint of the inclusion functor $U$ from the category $\mathbf{DCPO}$ of dcpos and Scott-continuous maps into the category $\mathbf{POS_{d}}$ of posets and Scott-continuous maps (not order-preserving maps). 
\[
\begin{tikzcd}
\mathbf{TOP_{0}}
\arrow[rr, "\D_{\cc} "] &&   \D \arrow[dd, "\Omega"] \\
\\
\mathbf{POS_{d}}\arrow[uu, "\Sigma"] \arrow[rr, "\Omega\circ \D_{\cc}\circ \Sigma" {name=F}, bend left = 15] && \mathbf{DCPO} \arrow[ll, "U"{name=G}, bend left=15 ]
\arrow[phantom, from=F, to=G, "\dashv" rotate=-90]
\end{tikzcd}
\]
Hence for each poset $P$, the pair $(\Omega(\D_{\cc}( \Sigma P)), e)$ is a dcpo-completion of $P$, where $e$ is the unit of the adjunction. 
\end{remark}

By composition with the inclusion of $\kk$ into $\bf TOP_0$,   $\kc$ can be seen as an endofunctor on ${\bf TOP_0}$, so it defines a monad. Then $\kc$ is an idempotent monad due to the universal property mentioned above. Its unit at $X$ is just $\eta_{X}$ also described above, and its multiplication $\mm_{X}$ at $X$ is the extension map of $\id_{X}$. Moreover, we know that $\eta_{\kc(X)}=\kc(\eta_{X})$, which is inverse to $\mm_{X}$; this implies $\kc(X)$ is homeomorphic to $\kc(\kc(X))$.
The (Eilenberg-Moore) algebras of $\kk_\cc$ are precisely the $\kk$-spaces (See for example \cite[Proposition 4.23, Corollary 4.24]{borceux2}).

Since each $\K$-category $\kk$ is a reflective, full and replete subcategory of ${\bf TOP_0}$,  the product of any family of $\kk$-spaces, as computed in  ${\bf TOP_0}$,  is a $\kk$-space and the same holds for equalisers. An easy consequence is that any $\K$-category is closed under products and equalisers, hence is complete. (See for example, Exercise~7 one page 92 in~\cite{maclane98}.)

The following slightly stronger statement also holds. 

\begin{lemma}
\label{equaliser:kspaces}
Let $X$ a $\kk$-space and $Y$ be a $T_0$ space, and let $f, g$ be continuous maps from $X$ to $Y$. If $h\colon Z\to X$ is an equaliser of $f$ and $g$, then $Z$ is also a $\kk$-space. 
\end{lemma}
\begin{proof}
\begin{figure}[h]
\centering
\begin{tikzcd}
 Z \arrow[r, "h"]  &   X \arrow[r, "g"', bend right =10]  \arrow[r, "f", bend left=10] & Y \arrow[r, hook, "\eta^{\s}_{Y}"] & Y^{s}  
\end{tikzcd}
\end{figure}
Let $Y^\s$ be the canonical sobrification of $Y$ and $\eta^{\s}_{Y}$ be the canonical topological embedding of $Y$ into $Y^{\s}$ sending $y\in Y$ to the closure of $\{y\}$.  Then we know that $h\colon Z\to X$ is also a equaliser of $\eta^{\s}_{Y}\circ f$ and $\eta^{\s}_{Y}\circ g$. Since $Y^\s$ is sober, hence a $\kk$-space, this implies that $Z$ is a $\kk$-space. 
\end{proof}

\begin{corollary}
\label{densedetermined}
Let $X$ be a $T_0$ space and let $f, g$ be continuous maps from $\kc(X)$ to  $T_0$ space $Y$. If $f$ and $g$ are equal on $X$, that is, $f\circ \eta_{X} = g\circ \eta_{X}$, then $f=g$. 
\end{corollary}
\begin{proof}
\begin{figure}[h]
\centering
\begin{tikzcd}
 X \arrow[r, "\eta_{X}"] \arrow[dr, "j"']  &   \kc(X)  \arrow[d, "\overline j", bend left = 10] \arrow[r, "f", bend left = 10]  \arrow[r, "g"', bend right=10] & Y  \\
 &  Z  \arrow[u, "h", bend left = 10] &
\end{tikzcd}
\end{figure}
Let $h\colon Z\to \kc(X)$ be the equaliser of $f$ and $g$. We know that there exists a unique $j\colon X\to Z$ such that $\eta_{X}=h\circ j$ since by assumption $\eta_{X}$ equalises $f$ and $g$.  Lemma~\ref{equaliser:kspaces} implies $Z$ is a $\kk$-space, so the map $j$ has a unique extension $\overline j\colon \kc(X)\to Z$ such that $j=\overline j\circ \eta_{X}$. We claim that $\overline j$ and $h$ are inverse to each other. First, notice that $h\circ \overline j\circ \eta_{X} = h\circ j= \eta_{X} =\id_{\kc (X)}\circ \eta_{X}$. This means that both the maps $h\circ \overline j$ and $\id_{\kc (X)}$ extend $\eta_{X}$ and hence it follows that $h\circ\overline j =  \id_{\kc(X)}$. This also implies that $h\circ \overline j\circ h=h$, which in turn implies that $\overline j\circ h = \id_{Z}$, since the map $h$, as an equaliser, is a monomorphism. Hence $\kc(X)$ and $Z$ are homeomorphic. Finally,  we conclude that $f=g$ since $h\circ f = h\circ g$ and $h$ is a homeomorphism. 
\end{proof}

As another corollary to Lemma~\ref{equaliser:kspaces}, every $\K$-category is also stable under \emph{retractions}. 

\begin{corollary}
\label{retracts:of:kspace}
Let $\kk$ be a $\K$-category, $Y$ be a $\kk$-space, and $X$ be a $T_0$ space. If $X$ is a retract of $Y$, then $X$ is also a $\kk$-space. 
\end{corollary}
\begin{proof}
\begin{figure}[h]
\centering
\begin{tikzcd}
 X \arrow[r, "s"]  &   Y \arrow[rr, "\id_{Y}"', bend right =13]  \arrow[r, "r", bend left=13] & X \arrow[r, "s", bend left=13] & Y 
\end{tikzcd}
\end{figure}
Let $r\colon Y\to X$ be a retraction and $s\colon X\to Y$ be its section. It is immediate that $s\colon X\to Y$ is the equaliser of $\id_{Y}$ and $s\circ r$. By Lemma~\ref{equaliser:kspaces} we conclude that $X$ is a $\kk$-space. 
\end{proof}

The following is \`a la \cite[Theorem 2.4]{heckmann96}. 

\begin{corollary}
\label{alahcekmann}
For a subspace $X$ of a $\kk$-space $Y$, consider the following statements:
\begin{enumerate}
\item The pair $(Y, e)$ is a $\kk$-completion of $X$, where $e$ is the subspace embedding of $X$ into $Y$.
\item For every $y\in Y$ and open subset $O$ of $Y$ with $y\in O$ there is some $x\in X$ with $x\in O$ and $x\leq y$;
\item For all opens $U$ and $V$ of $Y$, $X\cap U\subseteq V$ implies that $U\subseteq V$.
\end{enumerate}
Then $(1) \Rightarrow (2) \Leftrightarrow (3)$. 
\end{corollary}
\begin{proof}
The proof of the equivalence between $(2)$ and $(3)$ goes exactly as in that of~\cite[Theorem 2.4]{heckmann96}. 

We prove that $(1)$ implies $(2)$. 
To this end we assume that $(1)$ holds and $(2)$ does not. Then there exists some $y\in Y$ and open set $O\subseteq Y$ of $y$ such that $\da y\cap U\cap X = \emptyset$. This means that $X\subseteq Y\setminus (\da y\cap U)$. Now consider the map $\chi_{O\cap X}\colon X\to \mathbb S$. Since $\mathbb S$ is sober hence a $\kk$-space, the fact that both $\chi_{O}\colon Y\to \mathbb S$ and $\chi_{O\setminus \da y}\colon Y\to \mathbb S $ extend $\chi_{O\cap X}$ would contradict with the fact that $(Y, e)$ is a $\kk$-completion of $X$. 
\end{proof}

\begin{remark}\label{isobetweenxandkx}
The Item~$3$ of Corollary~\ref{alahcekmann} implies that for each open subset $W$ of $X$, there is one and only one open subset $O$ of $Y$, such that $W=O\cap X$, and $\chi_{O}\colon Y\to \mathbb S$ is the unique continuous map extending $\chi_{W}\colon X\to \mathbb S$. Moreover, the map $\lambda O. X\cap O$ is an order-isomorphism between the lattice of open sets of $Y$ and that of $X$. 
\end{remark}

\begin{proposition}
\label{kpreservesorder}
Let $X$ be a $T_{0}$ space, $K$ be a $\kk$-space and $f\colon X\to K$ be any continuous map. If we use $[X\to K]$ (resp., $[\kc(X) \to K]$) to denote the set of continuous functions from $X$ (resp., $\kc(X)$) to $K$,  then the extension map $f\mapsto \overline f\colon [X\to K]\to[\kc(X)\to K] $ is a bijection. Moreover, it is an order isomorphism. That is, for $f, g\colon X\to K$ with $f\leq g$, we have that $\overline f\leq \overline g$. The order mentioned here is the pointwise order between functions induced by the specialisation order of $K$. 
\end{proposition}
\begin{proof}
It is obvious that $(f\mapsto \overline f)$  is a bijection. We assume that it is not  an order isomorphism. 
Hence we know that  there exist functions $f$ and $g$ from $X$ to $K$ with $f\leq g$, but $\overline f\not\leq \overline g$. So there is some point $a\in \kc(X)$ such that $\overline f(a)\not\leq \overline g(a)$. This implies that $\overline f(a)$ is in the open set $U= K\setminus \da \overline g(a)$. By Item~2 of Corollary~\ref{alahcekmann} we have some $x\in X$ such that $\eta_{X}(x)\in \da a\cap {\overline f^{-1}}(U)$. Applying $\overline f$ to $\eta_{X}(x)$ yields that $f(x)= \overline f(\eta_{X}(x))\in U$. This implies that $f(x)\not\leq \overline g(a) $. However, this is impossible since $f(x)\leq g(x) =  \overline g(\eta_{X}(x)) \leq \overline g(a)$. (the last inequality comes from the fact that $\overline g$ is continuous hence order-preserving and that $\eta_{X}(x)\leq a$.)
\end{proof}

\section{The monad structure}
In this section, we prove that for each $\K$-category $\kk$, the composition $\kc \circ \vs$ gives rise to a monad over the category ${\bf TOP_{0}}$. First of all, it is easy to see that the assignment of the space $\kc(\vs X)$ to each space $X$,  and of the continuous function $\kc(\vs f)$ to each continuous function $f\colon X\to Y$  defines an endofunctor on ${\bf TOP_0}$, which we denote by~$\vk$.

The $\kk$-completion of $\vs X$, as a topological space, is homeomorphic to the smallest $\kk$-subspace  in $\V X$ containing $\vs X$. This is because \cite[Theorem 4.4]{keimel09a} tells us that we could complete $\vs X$ in any $\kk$-space that has (the homeomorphic copy of) $\vs X$ as a subspace. The space $\V X$, in our case, is a sober space~\cite[Proposition5.1]{heckmann96} hence a $\kk$-space and it contains $\vs X$ as a subspace.  Henceforth we identify $\vk X$ as a subspace of $\V X$.  

If $f\colon X\to Y$ is a continuous map, then $\V f\colon \V X\to \V Y$ also is continuous. Since $\V X, \V Y$ are sober spaces, $\V f$ is $\kk$-continuous by Item~4 of Definition~\ref{defk}. Thus, $\V f$ sends $\vk X$ into $\vk Y$. Moreover, $\vk f$ is defined to be the unique map $\kc(\vs f)$  such that $\kc(\vs f) \circ \eta_{\vs X} =\eta_{\vs Y}\circ \vs f$. Hence, $\vk f = \V f$ on $\vk X$, since for each simple valuation $s$, $\V f(s) =  \vs f(s)$. That is, for each $\mu\in \vk(X)$ and each open subset $U$ of $Y$, $\vk f (\mu)(U) = \V f (\mu)(U) = \mu(f^{-1}(U))$.

We summarise comments above in Figure~\ref{relation}, hooked arrows denote topological embeddings and tailed arrows denote subspace embeddings.

\begin{figure}[h]
\centering
\begin{tikzcd}
X    \arrow[d, "f"]  \arrow[r, hook, "\us_{X}"]     & \vs X     \arrow[d, "\vs f"]  \arrow[r, tail, "\eta_{\vs X}"]    & \vk X    \arrow[d, "\vk f"]  \arrow[r, tail, "e_{\vk X}"]     & \V X   \arrow[d, "\V f"]   \\
Y   \arrow[r, hook, "\us_{Y}"]     & \vs Y     \arrow[r, tail, "\eta_{\vs Y}"]    & \vk Y    \arrow[r, tail, "e_{\vk Y}"]     & \V Y 
\end{tikzcd}
\caption{ \label{relation} }
\end{figure}

Notice that $\V X$ is a topological cone containing $\vs X$ as a topological subcone. A natural question is whether the cone structure of $\V X$ restricts to that of $\vk X$. That is, if we take valuations from $\vk X$ and form finite sums and scalar multiples in $\V X$, do the results land back in $\vk X$? The following proposition answers this question in the positive. 

\begin{proposition}\label{subcone}
$\vk X$ is a subcone of $\V X$ for every topological space~$X$.  
\end{proposition}
\begin{proof}
Obviously the neutral valuation $(U\mapsto 0)$ is in $\vk X$. For any $\mu, \nu \in \vk X, r\in \R$, we prove that $\mu+ \nu, r\mu$,  as computed in $\V X$, are in $\vk X$. First, notice that $+\colon \V X\times \V X\to \V X$ is jointly continuous. Hence for each simple valuation $s\in \vs X$, $s+\_\colon\V X\to \V X$ is continuous, then $s'\mapsto s+ s'$ is $\kk$-continuous since $\V X$ is sober. Since $s'\mapsto s+ s'$ maps $\vs X$ into itself, that is, $s + (\vs X) \subseteq \vs X$,  it follows from Item~4  of Definition~\ref{defk} that $s+\clk(\vs X)\subseteq \clk(s+\vs X)\subseteq \clk(\vs X)$. Thus, $\nu\mapsto s+\nu$ maps $\vk X$ into $\vk X$. It follows that for each valuation $\mu\in \vk X$, $s\mapsto \mu+ s$ maps $\vs X$ into $\vk X$. Using the fact that $\mu+\_: \V X\to \V X$ is continuous, hence $\kk$-continuous, we conclude that $\mu+\nu\in \vk X$ for any $\mu, \nu\in \vk X$. Similar reasoning shows $r\mu\in \vk(X)$ for $r\in \R, \mu\in \vk X$, since $s\mapsto r\cdot s$ sends simple valuations to simple valuations for each $r\in\R$. 
\end{proof}

Since $\vk X$ has a canonical cone structure inherited from $\V X$, the following  statements make sense; when referring to the cone structure on $\vk X$, we always mean the cone structure inherited from $\V X$.

\begin{corollary}\label{llofvk}
$\vk X$ is a locally linear topological cone. 
\end{corollary}
\begin{proof}
Straightforward from the fact that $\V X$ is locally linear ant that $\vk X$ is a subcone of $\V X$.
\end{proof}

\begin{lemma}\label{kpreslinear}
Let $f\colon \vs X\to \vs Y$ be a continuous linear map, then $\kc (f)$ is a continuous linear map from $\vk X$ to $\vk Y$.
\end{lemma}
\begin{figure}[h]
\centering
\begin{tikzcd}
\vk X \arrow[rr, "\kc(f)"] && \vk Y \\
&&\\
\vs X \arrow[uu, "\eta_{\vs X}"] \arrow[rr, "f"] && \vs Y \arrow[uu, "\eta_{\vs Y}"]
\end{tikzcd}
\end{figure}

\begin{proof}
We denote $\kc(f)$ by $F$ and prove that $F$ is linear. Note that for a fixed simple valuation $s\in\vk X$, the two maps $$F_{1}\colon \mu\mapsto F(s) + F(\mu)\colon \vk X \to \vk Y~\text{and}~F_{2}\colon\mu \mapsto F(s+\mu)\colon \vk X \to \vk Y$$ are equal when $\mu$ is a simple valuation, since $f$ is linear. Hence they are equal, as they both equal $\kc(g)$, where $g\colon \vs X\to \vs Y$ by $g(t) = f(s)+f(t)$ for each simple valuation $t$.  Thus, the maps 
$$\nu \mapsto F(\nu) + F(\mu)\colon \vk X\to \vk Y~\text{and}~\nu \mapsto F(\nu+\mu)\colon \vk X\to \vk Y$$ are equal since they are unique continuous extensions of $(s \mapsto F(\eta_{\vs X}(s))+F(\mu) \colon \vs X \to \vk Y )$ and $(s \mapsto F(\eta_{\vs X}(s)+ \mu) \colon \vs X \to \vk Y)$, respectively. With a similar reasoning we can prove that $F$ is homogeneous.
\end{proof}

\begin{proposition}\label{dagger}
For topological spaces $X$ and $Y$, and for every continuous function $f\colon X\to \vk Y$, the map 
$$ f^\dagger_\kk \colon \mu \mapsto (U\mapsto \int_{x\in X} f(x) (U)d\mu  )\colon \vk X \to \vk Y$$
is well-defined and continuous.
\end{proposition}
\begin{proof}
Since $f^\dagger\colon \V X\to \V Y$ is continuous for any continuous map $f\colon X\to \vk Y$, it is sufficient to show that $f^\dagger$ sends $\vk X$ into $\vk Y$.

We start with showing that $f^\dagger$ sends $\vs X$ to $\vk Y$.
To this end, we let $\mu$ be $\Sigma_{i=1}^n r_i \delta_{x_i}$ and perform the following computation for each open subset $U$ of $Y$:
$$f^\dagger (\mu) (U) = \int_{x\in X} f(x) (U)d\mu = \Sigma_{i=1}^n r_i f(x_i)(U) =   (\Sigma_{i=1}^n r_i f(x_i)) (U).$$
This shows that $f^\dagger (\mu) =  \Sigma_{i=1}^n r_i f(x_i)$.  Since $f(x_i)\in  \vk Y$ for $i=1, ..., n$,  and $\vk Y$ is a subcone of $\V Y$ by Proposition~\ref{subcone}, we have $f^\dagger (\mu)=\Sigma_{i=1}^n r_i f(x_i)$ is in $\vk Y$. 

So we have that $f^\dagger|_{\vs X}$, the restriction of $f^\dagger$ on $\vs X$,  is continuous from $\vs X$ to $\vk Y$. By the universal property of $\kk$-completion we know that $f^\dagger|_{\vs X} $ has a unique extension $F$ from $\vk X$ to $\vk Y$. Since $\vk Y$ is a subspace of $\V Y$,  $F$ is also continuous from $\vk X$ to $\V Y$. We note that $f^\dagger|_{\vk X}$, the restriction of $f^\dagger$ on $\vk X$, is also continuous from $\vk X$ to $\V Y$, and moreover that $F$ and   $f^\dagger|_{\vk X}$  are equal on $\vs X$, so $F$ and  $f^\dagger|_{\vk X}$ are equal on $\vk X$ by Corollary~\ref{densedetermined}. Now we conclude that $f^\dagger$ sends $\vk X$ into $\vk Y$ because $F$ does. 
\end{proof}

\begin{theorem}
\label{vkisamonad}
For each $\K$-category $\kk$, $\vk$ is a monad on  ${\bf TOP_0}$. 
\end{theorem}
\begin{proof}
The unit of  $\vk$ sends each point $x\in X$ to the Dirac mass $\delta_x$ in $\vk X$, which we denote by $ {\uk}_{X} \colon X\to \vk X\colon x\mapsto \delta_{x}$. For each continuous function $f\colon X\to \vk Y$, we have an extension $f^\dagger_\kk\colon  \vk X \to \vk Y$ of $f$ guaranteed by Proposition~\ref{dagger}. To verify the monad laws one then follows the proof of \cite[Proposition 2.12]{goubault19}.

An easy consequence is that the multiplication of $\vk$ at $X$, in symbols $\mk_{X}$, is given by $({\id_{\vk X}})_{\kk}^{\dagger}$, as defined in Proposition~\ref{dagger}. 
\end{proof}

\begin{remark}
When $\kk$ is the category of sober spaces and continuous functions, the functor $\vk$ was shown to coincide with the \emph{point-continuous valuation} functor $\mathcal V_{\mathrm p}$ by Heckmann \cite{heckmann96}; that $\mathcal V_{\mathrm p}$ is a monad was shown in~\cite{goubault19}.
\end{remark}

\begin{remark}
We summarise different monads and their respective units and multiplications that we have so far.

\begin{center}
\begin{tabular}{ |c|c|c| } 
 \hline
monads & unit & multiplication \\ 
  \hline
 $\V$ & $ \delta$ & $\beta$ \\ 
  \hline
 $\vs$ & $\us$ & $\ms$ \\ 
 \hline
  $\kc$ & $\eta$ & $\mm$ \\ 
   \hline
     $\vk=\kc\circ \vs$ & $\uk$ & $\mk$ \\ 
 \hline
\end{tabular}
\end{center}
\end{remark}

In order to show that $\vk$ is a monad over ${\bf TOP_0}$, one can alternatively show that there exists a distributive law of $\vs$ over $\kc$. For this, we recall the definition of distributive laws between monads and a result due to Heckmann. 

\begin{definition}[Distributive law]
\label{distrilaw}
Let $(T_{1}, \eta_{1}, m_{1}), (T_{2}, \eta_{2}, m_{2})$ be monads over category $\bf C$. A distributive law of $T_{1}$ over $T_{2}$ is a natural transformation $\lambda\colon T_{1}T_{2}\to T_{2}T_{1}$ such that the diagrams in Figure~\ref{distributivelaws} commute. 
\begin{figure}[h]
\centering
\begin{tikzcd}
  & T_{1} \arrow[dl, "T_{1}\eta_{2}"'] \arrow[dr, "\eta_{2}T_{1}"]  & \\
T_{1}T_{2} \arrow[rr, "\lambda"] &  & T_{2}T_{1}\\
\end{tikzcd}
\begin{tikzcd}
  & T_{2} \arrow[dl, "\eta_{1}T_{2}"'] \arrow[dr, "T_{2}\eta_{1}"]  & \\
T_{1}T_{2} \arrow[rr, "\lambda"] &  & T_{2}T_{1}\\
\end{tikzcd}
\begin{tikzcd}
 T_{1}^{2}T_{2} \arrow[r, "T_{1}\lambda"]  \arrow[d, "m_{1}T_{2}"']   & T_{1}T_{2}T_{1} \arrow[r, "\lambda T_{1}"]  & T_{2}T_{1}^{2} \arrow[d, "T_{2}m_{1}"]  \\
T_{1}T_{2} \arrow[rr, "\lambda"] &  & T_{2}T_{1}\\
\end{tikzcd}
\begin{tikzcd}
 T_{1}T_{2}^{2} \arrow[r, "\lambda T_{2}"]  \arrow[d, "T_{1}m_{2}"']   & T_{2}T_{1}T_{2} \arrow[r, "T_{2}\lambda"]  & T_{2}^{2}T_{1} \arrow[d, "m_{2}T_{1}"]  \\
T_{1}T_{2} \arrow[rr, "\lambda"] &  & T_{2}T_{1}
\end{tikzcd}
\caption{\label{distributivelaws}}
\end{figure}
\end{definition}

\begin{theorem}{\rm \cite[Theorem 6.7]{heckmann96}}
\label{theoremfreeeofvf}
$\vs X$ is the free weakly locally convex cone over $X$ in ${\bf TOP_{0}}$. This means: for a $T_{0}$ topological space $X$, $\vs X$ is itself a weakly locally convex cone, and for every continuous function $f\colon X\to M$ from $X$ to a weakly locally convex cone $M$, there is a unique continuous linear function $\overline f \colon \vs X\to M$ with $\overline f\circ  \us_{X} = f$.  
\end{theorem}

Theorem~\ref{theoremfreeeofvf} and Proposition~\ref{subcone} allow us to define a distributive law of $\vs$ over $\kc$. Notice from Corollary~\ref{llofvk} that for each space $X$, $\vk X$ is a locally linear cone, in particular, a weakly locally convex cone. Then by Theorem~\ref{theoremfreeeofvf} the continuous map $\kc(\us_{X})\colon \kc(X) \to \vk X$ has a unique continuous extension 
$\overline{\kc(\us_{X})}\colon \vs \kc (X)\to \vk X$ that is linear, and that satisfies  $\overline{\kc(\us_{X})} \circ \us_{\kc(X)} = \kc(\us_{X})$ (See Figure~\ref{existenceoflambda}).  
\begin{figure}[h]
\centering
\begin{tikzcd}
\vs \kc(X) \arrow[dr, "\overline{ \kc(\us_{X})} "] &  \\
\kc(X) \arrow[u, "\us_{\kc(X)} "] \arrow[r, "\kc(\us_{X})"'] & \vk X \\
X \arrow[u, "\eta_{X}"]  \arrow[r, "\us_X"'] & \vs (X) \arrow[u, "\eta_{\vs X}"'] 
\end{tikzcd}
\caption{\label{existenceoflambda}}
\end{figure}
We claim that the collection of morphisms $\overline{ \kc(\delta_{X})}$ indexed by $X$ is the (unique) distributive law of $\vs$ over~$\kc$.

\begin{remark}
One may immediately notice that it would be more involved to show that there is a distributive law of $\vs$ over $\kc$, and then use that to prove $\vk =  \kk_C\circ \vs$ defines a monad, instead of appealing to Theorem~\ref{vkisamonad}. However, by doing so, we obtain additional information that is useful in Section~\ref{extensionsection}.
\end{remark}

\begin{theorem}
\label{existenceofdis}
The family $\overline{\kc(\us_{X})}\colon \vs\kc(X) \to \kc (\vs X)$, for $X$ in $\bf TOP_0$, of continuous linear maps form the (unique) distributive law of $\vs$ over $\kc$.
\end{theorem}
\begin{proof}
We first show that the maps $\overline{\kc(\us_{X})}$ form a natural transformation from $\vs\kc$ to $\kc\vs$. That is, for each continuous map $f\colon X\to Y$, the equation 
\begin{equation}
 \overline{ \kc(\us_{Y})} \circ \vs\kc(f) = \kc(\vs f)  \circ \overline{ \kc(\us_{X})}. \tag{$\star$}
 \end{equation}
holds.
\[
\begin{tikzcd}
\vs\kc(Y) \arrow[rr, "\overline{ \kc(\us_{Y})} "] &  & \kc(\vs Y)  \\
& \kc(Y) \arrow[ul, "\us_{\kc(Y)}"] \arrow[ur, "{\kc(\us_{Y})}"'] &\\
& \kc(X) \arrow[dl, "\us_{\kc(X)}"']  \arrow[u, "\kc(f)"] \arrow[dr, "\kc(\us_X)"]  & \\
 \vs\kc(X)  \arrow[uuu, "\vs{\kc(f)}"]  \arrow[rr, "\overline{ \kc(\us_{X})}"'] & & \kc(\vs X) \arrow[uuu, "\kc(\vs f)"']
\end{tikzcd}
\]
To this end we perform the following: 
\begin{align*}
&\overline{ \kc(\us_{Y})} \circ \vs\kc(f) \circ \us_{\kc (X)} \\ 
& =   \overline{ \kc(\us_{Y})}\circ \us_{\kc(Y)}\circ \kc(f)  & \text{naturality of $\us$  }\\
& =  \kc(\us_{Y})\circ \kc(f) & \text{definition of $\overline{ \kc(\us_{Y})}$}\\
& =  \kc(\vs f) \circ \kc(\us_{X})& \text{naturality of $\us$ and functoriality of $ \kc $}\\
& = \kc(\vs f)  \circ \overline{ \kc(\us_{X})}  \circ \us_{\kc X}. &\text{definition of $\overline{ \kc(\us_{X})}$ }
\end{align*}
Notice that both sides of the Equation~$(\star)$ are continuous linear maps by Lemma~\ref{kpreslinear} and the fact that $\vs f$ is linear for all continuous~$f$, by the freeness of  
$\vs\kc(X)$ over $\kc(X)$ (Theorem~\ref{theoremfreeeofvf}) we conclude that Equation~$(\star)$ holds. For each $X$ we denote $\overline{ \kc(\us_{X})}$ by $\lambda_{X}$, and hence $\lambda$ is indeed a natural transformation from $\vs\kc$ to $\kc \vs$.
\[
\begin{tikzcd}
& &  \vs X \arrow[ddll, "\vs\eta_{X}"'] \arrow[dd, "\eta_{\vs X}"] &  & \\
& & & & \\
\vs \kc(X) \arrow[rr, "\lambda_{X}" description] & & \kc(\vs X) & & X\arrow[uull, "\us_{X}"'] \arrow[ddll, "\eta_{X}"]\\
& & & & \\
& & \kc(X)\arrow[uull, "\us_{\kc(X)}"] \arrow[uu, "\kc(\us_{X})"'] & &
\end{tikzcd}
\]

To verify the commutativity of the first digram of Definition~\ref{distrilaw} in our setting, we show that 
$ \lambda_{X}\circ  \vs\eta_{X}= \eta_{\vs X}$. We notice that
\begin{align*}
&\lambda_{X} \circ  \vs\eta_{X} \circ \us_{X}  \\ 
& = \lambda_{X} \circ \us_{\kc(X)}\circ \eta_{X} &\text{ naturality of $\us$  }  \\
& =    \kc(\us_{X})   \circ \eta_{X}      & \text{definition of $\lambda_{X}$}\\
& = \eta_{\vs X} \circ \us_{X} & \text{definition of $ \kc $}
\end{align*}
and use the fact that $ \lambda_{X}\circ  \vs\eta_{X}$ and $ \eta_{\vs X}$ are continuous linear maps and  freeness of $\vs(X)$ over $X$ to conclude that  
\begin{equation} 
 \lambda_{X}\circ  \vs\eta_{X} = \eta_{\vs X} \tag{$\bullet$}.
\end{equation}
The definition of~$\lambda$ trivially implies the commutativity of the second digram in our setting, and the uniqueness of $\overline{ \kc(\us_{X})}$ guarantees the uniqueness of such a distributive law $\lambda$. 

We now verify the commutativity of the third digram in Figure~\ref{distributivelaws} in our setting. That is, we need to prove that
 \begin{equation}
 \kc \ms \circ \lambda_{\vs}\circ \vs \lambda = \lambda\circ \ms_{\kc}.\tag{$\ast$}
 \end{equation}
To this end, we notice that $\vk X$ is a weakly locally convex cone by Corollary~\ref{llofvk}, and it then follows from  Theorem~\ref{theoremfreeeofvf}  that  there exists a unique continuous linear map $h_{X}\colon \vs\kc(\vs X)\to \vk X$  such that $h_{X}\circ \us_{\vk X} = \id_{\vk X}$. We claim that 
\begin{equation}
\lambda_{X}\circ \ms_{\kc(X)} = h_{X}\circ \vs\lambda_{X} \tag{$1$}
\end{equation}
as illustrated in the following:

\[
\begin{tikzcd}
\vs^{2} \kc \arrow[rr, "\vs \lambda"] \arrow[ddd, "\ms_{\kc}"']  & & \vs\kc\vs \arrow[rr,  "\lambda_{\vs}"]    \arrow[ddd, bend right = 25 , "h" description]  & & \kc\vs^{2} \arrow[ddd, "\kc\ms"]  \\
 &&&&\\
 &&&&\\ 
\vs\kc \arrow[rr, "\lambda"]  && \kc\vs \arrow[rr, "\id_{\kc\vs}", equal] \arrow[uuurr, "\kc\us_{\vs} " description]   \arrow[uuu, bend right = 25, "\us_{\kc\vs}"' description ]  && \kc\vs
\end{tikzcd}
\]
 For this, we take any simple valuation $\Sigma_{i=1}^{n}r_{i}\delta_{\mu_{i}}$ in $\vs^{2}\kc(X)$, where for each $i$, $\mu_{i}$ is some simple valuation on $\kc(X)$, and perform the following: 
\begin{align*}
& h_{X}\circ \vs\lambda_{X}(\Sigma_{i=1}^{n}r_{i}\delta_{\mu_{i}})\\
& =h_{X}(\Sigma_{i=1}^{n}r_{i}\delta_{\lambda_{X}(\mu_{i}) }  ) & \text{actions of $\vs\lambda_{X}$} \\
& = \Sigma_{i=1}^{n}r_{i} h_{X}(\delta_{\lambda_{X}(\mu_{i})}) & \text{$h_{X}$ is linear}\\
& = \Sigma_{i=1}^{n}r_{i} \lambda_{X}( \mu_{i}) &  h_{X}\circ \us_{\vk X} = \id_{\vk X}\\
& = \lambda_{X}(\Sigma_{i=1}^{n} r_{i}\mu_{i}  )  &\text{ $\lambda_{X}$ is linear }  \\
& = \lambda_{X}\circ \ms_{\kc(X)} (\Sigma_{i=1}^{n}r_{i}\delta_{\mu_{i}})& \text{definition of $\ms_{\kc(X)}$}
\end{align*}
Next, we show that 
\begin{equation}
\kc(\ms_{X}) \circ\lambda_{\vs X } =  h_{X}. \tag{$2$}
\end{equation} 
By Lemma~\ref{kpreslinear} both sides of this equation are continuous linear maps, so we only need to show that 
$$\kc(\ms_{X}) \circ\lambda_{\vs X}\circ \us_{\vk X} =  h_{X}\circ \us_{\vk X}$$ 
and conclude with Theorem~\ref{theoremfreeeofvf}. Indeed, the right side of the equation is just $\id_{\vk X}$ by the definition of $h_{X}$. For the left side, one notices that $\kc(\ms_{X}) \circ\lambda_{\vs X}\circ \us_{\vk X} =\kc(\ms_{X})\circ \kc(\us_{\vs X})$ by definition of $\lambda_{\vs X}$. But $\kk_C(\ms_X)\circ\kk_C(\us_{\vs X}) = \kk_C(\ms_X\circ\us_{\vs X})$ as $\kk_C$ is a functor,  and  $\ms_X\circ\us_{\vs X} = \id_{\vs X}$, so $\kc(\ms_{X}\circ \us_{\vk X}) = \kc(\id_{\vs X}) = \id_{\vk X}$, again since $\kc$ is a functor. Then Equation~$(\ast)$ follows directly from Equations~$(1)$ and $(2)$.

\[
\begin{tikzcd}
\vs \kc^{2} \arrow[rr, "\lambda_{\kc}"] \arrow[ddd, "\vs{\mm}", bend left = 20] & & \kc\vs\kc \arrow[rr,  "\kc \lambda"]   & & \kc^{2}\vs \arrow[ddd, "\mm_{\vs}", bend left =20 ]  \\
 &&&&\\
 &&&&\\ 
\vs\kc \arrow[rr, "\lambda"] \arrow[uuu, "\vs{\kc \eta}", bend left =20]   && \kc\vs \arrow[rr, "\id_{\kc\vs}", equal] \arrow[uuu, "\kc\vs\eta"']    && \kc\vs \arrow[uuu, "{\kc\eta_{\vs}}", bend left =20]
\end{tikzcd}
\]

Finally, to show the commutativity of the last diagram of the distributive law in our setting, we first notice that $\vs\kc\eta$ is inverse to $\vs\mm$ and that $\kc\eta_{\vs} $ is inverse to $\mm_{\vs}$ since $\kc$ is idempotent. Then we only need to show that 
\begin{equation}
\kc\lambda\circ \lambda_{\kc}\circ\vs\kc\eta = \kc\eta_{\vs}\circ \lambda. \tag{$**$}
\end{equation} To show Equation~$(**)$ holds, it suffices to prove that $\lambda_{\kc} \circ \vs\kc\eta = \kc\vs\eta \circ \lambda$ and that $\kc\lambda \circ\kc\vs\eta =\kc\eta_{\vs}$. 
These two equations hold trivially; the first one follows from the naturality of~$\lambda$, and the second follows from Equation~($\bullet$), which we proved above, and the functoriality of~$\kc$.
\end{proof}

The following is then a straightforward application of Beck's theorem of composing monads~\cite{beck69}, knowing that $\lambda$ is a distributive law of $\vs$ over~$\kk_{\cc}$. 

\begin{corollary}
The functor $\vk = \kc\circ \vs$ is a monad whose unit  is $\eta_{\vs}\us$, and whose multiplication is $\mm_{\vs}\circ \kc^{2}\ms \circ \kc\lambda_{\vs} = \kc\ms \circ \mm_{\vs^{2}} \circ \kc\lambda_{\vs}$. \hfill $\Box$
\end{corollary}

\section{The Eilenberg-Moore algebras of $\vk$}

Recall that an \emph{(Eilenberg-Moore) algebra of a monad} $\mathcal T$ over category ${\bf C}$ is a pair $(A, \alpha)$, where $A$ is an object in ${\bf C}$ and $\alpha_A\colon \mathcal T(A)\to A$ is a morphism of~${\bf C}$, called the \emph{structure map}, such that $\alpha_A \circ \eta_A=\id_A$ and $\alpha_A\circ\mu_A=\alpha_A\circ \mathcal T\alpha_A$. We often simply call such a pair $(A, \alpha_A)$ a $\mathcal T$-algebra. Given $\mathcal T$-algebras $(A, \alpha_A)$ and  $(B, \alpha_B)$, a morphism $f\colon A\to B$ of~${\bf C}$ is a \emph{$\mathcal T$-algebra morphism} if $f\circ \alpha_A=\mathcal Tf\circ \alpha_B$. 

In order to summarise what the $\vs$-algebras and $\vs$-algebra morphisms are, we recall the definition of a barycentre of a valuation.

\begin{definition}
  \label{defn:bary:choquet}
  Let $C$ be a semitopological cone, and $\nu$ be a continuous
  valuation on $C$.  A \emph{barycentre} of $\nu$ is any point
  $b_\nu \in C$ such that, for every linear lower semicontinuous map
  $\Lambda \colon C \to \real$,
  $\Lambda (b_\nu) = \int \Lambda ~d \nu$.
\end{definition}

\begin{example}
Let $X$ be a semitopological cone and $\Sigma_{i=1}^n r_i\delta_{x_i}$ be a simple valuation on $X$. Then $\Sigma_{i=1}^n r_i {x_i}$ is a barycentre of $\Sigma_{i=1}^n r_i\delta_{x_i}$. 
\end{example}

\begin{theorem}{\rm \cite[Theorem 5.9, Theorem 5.10]{goubault19}}
\label{theoremvf}
Let $X$ be a $T_0$ space, and $\alpha_X: \vs X\to X$ be a continuous map.
\begin{enumerate}
\item 
If $(X, \alpha_X)$ is a $\vs$-algebra, then $X$ is a weakly locally convex topological cone: For $x, y\in X$ and $r\in \mathbb R_+$, $x+y$ is defined as $\alpha_X(\delta_x+\delta_y)$ and $r\cdot x$ is defined as $\alpha_X(r\delta_x)$. The map $\alpha_X$ is then the barycentre map which sends each simple valuation $\Sigma_{i=1}^n r_i \delta_{x_i}$  to its barycentre $\Sigma_{i=1}^n r_i{x_i}$.
\item
Conversely, for every weakly locally convex topological cone~$C$, there exists a (unique) continuous linear map~$\alpha_C$ from $\vs C$ to $C$, sending each simple valuation $\Sigma_{i=1}^n r_i \delta_{x_i}$  to its barycentre $\Sigma_{i=1}^n r_i{x_i}$, and the pair $(C, \alpha_C)$ is a $\vs$-algebra. 
\item
Moreover, $\vs$-morphisms are precisely the continuous linear maps between $\vs$-algebras. 
\end{enumerate}
\end{theorem}

From Corollary~\ref{llofvk} we know that $\vk X$ is a locally linear topological cone for each $T_0$ space $X$. Then the following proposition is not hard to obtain.

\begin{proposition}\label{prop:vkalgcones}
Let $(X, \rho_X)$ be a $\vk$-algebra. Then $X$ is a weakly locally convex topological cone with $+$ defined by $x + y = \rho_X(\delta_x + \delta_y)$, and scalar multiplication $\cdot$ defined by $r \cdot x = \rho_X(r \delta_x)$ for $r \in \mathbb R_+$ and $x, y\in X$. 

Moreover $X$ is also a $\kk$-space, and the structure map $\rho_{X}$ maps each valuation $\mu\in \vk X$ to one of its barycentres. 
 \end{proposition}
\begin{proof}
The proof that $X$ is a topological cone is similar to that of Lemma~4.6 and Proposition~4.9 in \cite{goubault19}. In particular, we know that $\rho_{X}$ is a continuous linear retraction since it is a structure map. Hence $X$ is a weakly locally convex cone by Lemma~\ref{retractionofweaklylc}.  Since $\rho_X$ is a retraction and the category~$\kk$ is stable under retractions by Corollary~\ref{retracts:of:kspace}, we conclude that $X$ is a $\kk$-space.
We next prove that $\rho_{X}$ sends each valuation $\vk X$ to one of its barycentres. Assume that $\Lambda\colon X\to \real$ is a lower semicontinuous linear map. We consider two maps from $\vk X$ to $\real$: $\mu\mapsto \Lambda(\rho_{X}(\mu))$ and $\mu\mapsto \int \Lambda~d\mu$. It is straightforward to see that these two maps are continuous and they agree on $\vs X$, hence they also agree on $\vk X$ by Corollary~\ref{densedetermined}. So we have $\Lambda\circ\rho_{X}= \int \Lambda d\_$. 
\end{proof}

We call a topological cone $C$ which is also a $\kk$-space a \emph{$\kk$-cone}. Proposition~\ref{prop:vkalgcones} shows that every $\vk$-algebra is a weakly locally convex $\kk$-cone. The following proposition shows the converse also holds. That is,  each weakly locally convex $\kk$-cone is the underlying space of a $\vk$-algebra.

\begin{proposition}
\label{whatbetadoes}
For every weakly locally convex topological $\kk$-cone~$C$, there exists a (unique) continuous linear map~$\rho_C$ from $\vk C$ to $C$, sending each continuous valuation in $\vk C$ to one of its barycentres, in particular, sending each simple valuation $\Sigma_{i=1}^n r_i \delta_{x_i}$  to $\Sigma_{i=1}^n r_i{x_i}$. Moreover, the pair $(C, \rho_C)$ is a $\vk$-algebra. 
\end{proposition}
\begin{figure}[h]
\centering
\begin{tikzcd}
\vk C= \kc(\vs C)  \arrow[dr, "\rho_{C}=\overline{\alpha_{C}}"] \\
 \vs C \arrow[r, "\alpha_{C}"']  \arrow[u, hook, "\eta_{\vs C}"] &  C   
\end{tikzcd}
\end{figure}
\begin{proof}
Since $C$ is a weakly locally convex $\kk$-cone, it is in particular a $\vs$-algebra, hence from Theorem~\ref{theoremvf} we know that the (unique) map 
$\alpha_{C}\colon \vs C\to C$ that takes each simple valuation $\Sigma_{i=1}^n r_i \delta_{x_i}$ to its barycentre $\Sigma_{i=1}^n r_i{x_i}$ is continuous and linear. Since $\kc$ is a reflector and $C$ is a $\kk$-space, the map $\alpha_{C}$ has a unique continuous extension $\overline{\alpha_{C}}$ from $\vk C$ to $C$, which we claim is  the desired~$\rho_{C}$. 

The map $\rho_{C} = \overline{\alpha_{C}}$ is obviously continuous. To see that it is linear, for each continuous valuation $\mu\in \vk C$ we define $$F_{\mu}\colon \vk C\to C\colon \nu\mapsto \rho_{C}(\nu)+\rho_{C}(\mu)$$ and $$G_{\mu}\colon \vk C\to C\colon \nu\mapsto \rho_{C}(\nu+\mu).$$
It is easy to see that both $F_{\mu}$ and $G_{\mu}$ are continuous for each $\mu$. Notice that $\eta_{\vs C}$ is the space embedding of $\vs C$ into $\vk C$, so for simple valuations $s$ and $t$ we have:
\begin{align*}
F_{s}(t) &= \rho_{C}(s)+\rho_{C}(t)   &\text{definition of $F_{s}$} \\
&= \alpha_{C}(s)+\alpha_{C}(t) & \text{$\rho_{C}$ is an extension of $\alpha_{C}$ } \\
&= \alpha_{C}(s+t)  &\text{$\alpha_{C}$ is linear}\\
&= \rho_{C}(s+t) & \text{$\rho_{C}$ is an extension of $\alpha_{C}$}\\ 
&=G_{s}(t)  & \text{definition of $G_{s}$} 
\end{align*}
So we have that $F_{s}$ and $G_{s}$ are equal on $\vs C$ hence they are equal on $\vk C$ by  Corollary~\ref{densedetermined}. This means for each continuous valuation $\nu\in \vk C$, $F_{s}(\nu)=G_{s}(\nu)$. Notice that $F_{s}(\nu) = F_{\nu}(s)$ and $G_{s}(\nu) = G_{\nu}(s)$ for any simple valuation $s$ and continuous valuation $\nu\in \vk C$. We know that $F_{\nu}$ are  $G_{\nu}$ are equal on $\vs(C)$ for each $\nu$. Hence by using Corollary~\ref{densedetermined} again, they are equal on $\vk C$. That is, for each $\nu, \mu\in \vk C$, $F_{\nu}(\mu)=G_{\nu}(\mu)$, which is just equivalent to saying that $\rho_{C}(\nu)+ \rho_{C}(\mu) = \rho_{C}(\nu+\mu)$.

To prove that $\rho_{C}(r\mu) = r\rho_{C}(\mu)$ for $r\in \mathbb R_+, \mu\in \vk C$, one notes that both sides of the equation are continuous in $\mu$, and they are equal on simple valuations. Hence they are equal.

Now we prove that $\rho_{C}$ sends continuous valuations in $\vk C$ to one of their barycentres. Assume that $\Lambda\colon C\to \real$ is a lower semicontinuous linear map. We define 
$$ H\colon \vk C \to \real \colon \mu\mapsto \int f~d\mu$$
and 
$$ J\colon \vk C \to \real \colon \mu\mapsto f(\rho_{C}(\mu)).$$
It is straightforward to show $H$ and $J$ are continuous, and they are equal on simple valuations because $\rho_{C}$ is a linear extension of $\alpha_{C}$. Hence they are equal on $\vk C$ by~Corollary~\ref{densedetermined}.

The uniqueness of $\rho_{C}$ follows from the uniqueness of $\alpha_{C}$.

Finally, we prove that the pair $(C, \rho_{C})$ is a $\vk$-algebra. One easily sees that $\rho_{C}\circ \uk_{ C} = \id_{C}$ since $\rho_{C}$ extends $\alpha_{C}$. To prove that 
$\rho_{C}\circ \vk \rho_{C}  = \rho_{C}\circ \mk_{C}$ one notes that both sides of the equation are linear continuous maps, hence we only need to prove that $\rho_{C}\circ \vk \rho_{C}(\delta_{\nu})  = \rho_{C}\circ \mk_{C} (\delta_{\nu})$ for every $\nu \in \vk C$. Indeed we know that these are equal since they both are just $\rho_{C}(\nu)$ via easy calculations. 
\end{proof}

\begin{proposition}
Let $(C_{1}, \rho_{1}),  (C_{2}, \rho_{2})$ be two $\vk$-algebras. Then a continuous function $f\colon C_{1}\to C_{2}$ is a $\vk$-algebra morphism if and only if it is linear. 
\end{proposition}
\begin{proof}
For the ``if'' direction, we show that for every continuous linear map $f$, $\rho_{2}\circ \vk f = f\circ \rho_{1}$.  By Corollary~\ref{densedetermined} we only need to show that the maps are equal on simple valuations. Indeed, 
\begin{align*}
&\rho_{2}\circ \vk f(\Sigma_{i=1}^{n}r_{i}\delta_{x_{i}}) \\
&= \rho_{2}( \Sigma_{i=1}^{n}r_{i}\delta_{f(x_{i})}) &\text{$\vk f = \V f$ on $\vk C_{1}$}\\
&= \Sigma_{i=1}^{n}r_{i}{f(x_{i})} &  \text{Proposition~\ref{whatbetadoes}}\\
&=  f(\Sigma_{i=1}^{n}r_{i}{x_{i}} ) & \text{$f$ is linear}\\
&= f\circ \rho_{1}( \Sigma_{i=1}^{n}r_{i}\delta_{x_{i}} ). &\text{Proposition~\ref{whatbetadoes}   }
\end{align*}

For  the ``only if'' direction, we take $x_{i}\in C_{1}, r_{i}\in \mathrm{R_{+}}, i=1, ...,n$ and calculate 
\begin{align*}
& f(\Sigma_{i=1}^{n}r_{i}{x_{i}}) \\
&= f(\rho_{1}( \Sigma_{i=1}^{n}r_{i}{\delta_{x_{i}}})) &  \text{Proposition~\ref{whatbetadoes}}\\
&=  \rho_{2}(\vk f ( \Sigma_{i=1}^{n}r_{i}{\delta_{x_{i}}}))  & \text{$f$ is a $\vk$-algebra morphism }\\
&= \rho_{2}( \Sigma_{i=1}^{n}r_{i}\delta_{f(x_{i})}   ) & \text{$\vk f$ is linear}\\
&= \Sigma_{i=1}^{n}r_{i}f(x_{i}) . &\text{Proposition~\ref{whatbetadoes}   }
\end{align*}
\end{proof}

Summarising the above results, we have the following.

\begin{theorem}
The category ${\bf TOP_{0}^{\vk}}$ of $\vk$-algebras and $\vk$-algebra morphisms is isomorphic to the category of weakly locally convex $\kk$-cones and continuous linear maps between them. \hfill $\Box$
\end{theorem}

The following theorem which generalises \cite[Theorem 6.8]{heckmann96} is then immediate, and holds for general categorical reasons.
\begin{theorem}
$\vk X$ is the free weakly locally convex $\kk$-cone over $X$. That is, $\vk X$ is a weakly locally convex $\kk$-cone, and for every continuous function $f\colon X\to C$ from $X$ to a weakly locally convex $\kk$-cone $C$, there is a unique continuous linear map $\overline f \colon \vk X\to C$ with $\overline f(\delta_{x})=f(x)$ for all $x\in X$.\hfill $\Box$
\end{theorem}

\section{Extensions of lower semicontinuous linear maps}
\label{extensionsection}

In this section, we show that the cone structure on each weakly locally convex cone $C$ can be extended to $\kc(C)$ making the latter a weakly locally convex $\kk$-cone. Moreover, each lower semicontinuous linear map $\Lambda\colon C\to \real$ admits a unique lower semicontinuous linear extension $\overline \Lambda\colon \kc(C)\to \real$ such that $\overline\Lambda\circ \eta_{C}=\Lambda$.

\begin{proposition}
\label{kcisacone}
Let $C$ be a weakly locally convex topological cone, $\alpha_{C}$ be the structure map of the $\vs$-algebra at $C$, and $\us_{C}$ be the unit of $\vs$ at $C$. Then $\kc(C)$ is a weakly locally convex topological $\kk$-cone: For $x, y\in \kc(C)$, $r\in \mathbb R_{+}$, $x+y$ is defined as $\kc(\alpha_{C})(\kc(\us_{C})(x)+\kc(\us_{C})(y))$ and $r\cdot x$ is defined as  $\kc(\alpha_{C})(r\kc(\us_{C})(x))$. Moreover, the topological embedding $\eta_{C}$ of $C$ into $\kc(C)$ is linear with respect to this cone structure on $\kc(C)$.
\end{proposition}
\begin{figure}[h]
\centering
\begin{tikzcd} 
\vs \kc(C) \arrow[rr, " \lambda_{C}"]  &&\kc(\vs C) \arrow[ddd, "\kc(\alpha_{C})", bend left = 20] && \vs C \arrow[ll, "\eta_{\vs C} "]  \arrow[ddd, "\alpha_{C}", bend left = 20] \\
&\\
&\\
&&\kc(C) \arrow[uuu, "\kc(\us_{C})" description , bend left =20] \arrow[uuull, "\us_{\kc(C)} "] && C\arrow[ll, "\eta_{C}"] \arrow[uuu, "\us", bend left = 20]
\end{tikzcd}
\end{figure} 
\begin{proof}
Theorem~\ref{existenceofdis} implies there is a unique distributive law $\lambda$, consisting of linear maps,  of $\vs$ over $\kc$, and then Beck's lifting theorem states that $(\kc(C), \kc(\alpha_{C})\circ \lambda_{C})$ is a $\vs$-algebra (see the Proposition on Page 122 in \cite{beck69}). Then, Theorem~\ref{theoremvf} implies $\kc(C)$ is a weakly locally convex cone in which $x+y$ is defined as 
$\kc(\alpha_{C})( \lambda_{C} ( \delta_{x}+\delta_{y}) )$ and $r x $ is defined as $\kc(\alpha_{C})(\lambda_{C} ( r \delta_{x}))$. Moreover, we note that $\lambda_{C}$ is linear and $\lambda_{C}(\delta_{x}) = \kc(\us_{C})(x) $ for all $x\in \kc(C)$, so $\kc(\alpha_{C})( \lambda_{C} ( \delta_{x}+\delta_{y}) ) = \kc(\alpha_{C})(\kc(\us_{C})(x)+\kc(\us_{C})(y)) $ and $ \kc(\alpha_{C})(\lambda_{C} ( r \delta_{x})) = \kc(\alpha_{C})(r\kc(\us_{C})(x))$. 

To prove that $\eta_{C}$ is linear, let $a, b\in C$, and calculate:
\begin{align*}
&\eta_{C}( a + b ) \\
&= \eta_{C} (\alpha_{C}(\delta_{a}+\delta_{b})  )  & \text{ $\alpha_{C}$ is a structure map } \\ 
&= \kc(\alpha_{C})( \eta_{\vs C} (\delta_{a}+ \delta_{b})   )  &\text{ definition of $ \kc$ }\\
&= \kc(\alpha_{C}) ( \eta_{\vs C} (\delta_{a}) +  \eta_{\vs C}(\delta_{b})  )  & \text{ $\eta_{\vs C}$ is linear }\\
&= \kc(\alpha_{C}) ( \kc(\us_{C})(\eta_{C}(a ) ) +   \kc(\us_{C})(\eta_{C}(a ) )    ) &\text{definition of $\kc$ }\\
&= \eta_{C}(a) + \eta_{C}(b). &\text{definition of  $\eta_{C}(a) + \eta_{C}(b)$}
\end{align*}
The equation $\eta_{C}(ra)= r\eta_{C}(a)$ can be proved similarly, for  $r\in \mathbb R_{+}, a\in C$. 
\end{proof}

Let $C$ be a weakly locally convex cone and $f\colon C\to \real$ be lower semicontinuous. Since  $\real$ is a sober space, hence a $\kk$-space, and $\kc$ is a reflector, there is a unique continuous extension $\overline f\colon \kc(C)\to \real$. Moreover, since $C$ is weakly locally convex, $\kc(C)$ is a weakly locally convex $\kk$-cone by Proposition~\ref{kcisacone}. We can show more:

\begin{theorem}
\label{extensionoflinearmaps}
Let $C$ be a weakly locally convex cone. Then a lower semicontinuous map $f \colon C\to \real$ is homogeneous  (resp., superadditive, subadditive) if and only if its extension $\overline f\colon \kc(C)\to \real $ is homogeneous  (resp., superadditive, subadditive). As a result, $f$ is superlinear (resp., sublinear, linear) if and only if $\overline f$ is, respectively. 
\end{theorem}
\begin{proof}
We prove the case that $f$ is superadditive, that is, $f(a+b)\geq f(a)+f(b)$ for all $a, b\in C$. We need to show that for any fixed $c, d\in \kc(C)$, $\overline f(c+d)\geq \overline f(c)+\overline f(d)$. For fixed $a\in C$, it follows from the linearity of $\eta_{C}$ that $\lambda x. \overline f(\eta_{C}(a)+x)\colon \kc(C)\to K$ is the extension of $\lambda x. f(a+x)\colon C\to K$ and that $\lambda x.  f(a)+ \overline f(x)\colon \kc(C)\to K$ is the extension of $\lambda x.  f(a)+ f(x)\colon C\to K$. By superadditivity of $f$ we know that $\lambda x. f(a+x)\geq \lambda x.  f(a)+ f(x)$, hence it follows from Proposition~\ref{kpreservesorder} that $\lambda x. \overline f(\eta_{C}(a)+x)\geq \lambda x.  f(a)+  \overline f(x)$. In particular, we know that  $\overline f(\eta_{C}(a)+d)\geq f(a)+  \overline f(d)$ for any $a\in C$. This means that  $\lambda a. \overline f(\eta_{C}(a)+d)\geq \lambda a.  f(a)+  \overline f(d)$. By applying Proposition~\ref{kpreservesorder} again with the same reasoning we obtain that $\overline f(x+d)\geq \overline f(x)+\overline f(d)$ for any $x\in \kc(C)$. Hence, we have  $\overline f(c+d)\geq \overline f(c)+\overline f(d)$. 

The rest can be proved similarly. 
\end{proof}

The following result due to Keimel states that there is an order isomorphism between the family of convex open subsets of a topological cone $C$, ordered by set containment, and the family of  superlinear functionals on $C$, ordered pointwise,  via the so-called \emph{upper Minkowski functionals}.
\begin{definition}
Let $A$ be a subset of a cone $C$, the \emph{upper Minkowski functional} $M^{A}$ of $A$ is defined as $$ M^{A}(x)= \sup\{r \in \mathbb R_{+} \mid x\in rA\}. $$
\end{definition}

\begin{lemma}{\rm \cite[Proposition 7.4, Proposition 7.6]{keimel08}}
Let $C$ be a semitopological cone. 
\begin{itemize}
\item Assigning the upper Minkowski functional $M^{U}$ to each proper open subset $U$ of  $C$ induces an order isomorphism between the poset $\mathcal O(C)$ of proper open subsets and the poset of homogeneous lower semicontinuous functionals on $C$. 
\item Assigning the upper Minkowski functional $M^{U}$ to each proper open convex subset $U$ of  $C$ induces an order isomorphism between the poset $\mathcal O_{c}(C)$ of proper open convex subsets and the poset of lower semicontinuous superlinear functionals on $C$. 
\item In both cases, the inverse assignment sends each lower semicontinuous (superlinear) functional $f$ to the open $f^{-1}((1, +\infty ] )$. 
\end{itemize}
\end{lemma}

We have the following similar result with proper open half-spaces in lieu of proper open convex subsets and lower semicontinuous linear functionals in lieu of lower semicontinuous superlinear functionals. The proof is similar to \cite[Proposition 7.6]{keimel08}. 

\begin{lemma}
\label{inlieu}
Let $C$ be a semitopological cone. Assigning the upper Minkowski functional $M^{U}$ to every proper open half-space $U$ of  $C$ induces an order isomorphism between the poset of proper open half-spaces of $C$ and the poset of lower semicontinuous linear functionals on $C$. The inverse assignment sends each lower semicontinuous linear functional $f$ to the open $f^{-1}((1, +\infty ] )$.  \hfill $\Box$
\end{lemma}

Finally, we have the following result as a corollary to Theorem~\ref{extensionoflinearmaps}. 

\begin{theorem}
Let $C$ be a topological cone. Then $C$ is locally convex (resp., locally linear, weakly locally convex) if and and only if $\kc(C)$ is a locally convex (resp., locally linear, weakly locally convex) $\kk$-cone. 
\end{theorem}
\begin{proof}
The ``if'' direction comes from the fact that $\eta_{C}$ is linear by Proposition~\ref{kcisacone}. We prove the ``only if'' direction for the case when $C$ is locally convex. To this end, we assume that $U$ is an open subset of $\kc (C)$, and $x\in U$. Then $\eta_{C}^{-1}(U)$ is an open subset of $C$. By Corollary~\ref{alahcekmann} we have some point $a\in \eta_{C}^{-1}(U)$ with $\eta_{C}(a)\leq x$. Since $C$ is locally convex, there exists an open convex subset $W$ of $C$ such that $a\in W\subseteq \eta^{-1}(U)$. Hence the upper Minkowski functional $M^{W}$ of $W$ is a lower semicontinuous  superlinear map from $C$ to $\real$. By Theorem~\ref{extensionoflinearmaps}, it has an extension $\overline{ M^{W} }\colon \kc(C)\to \real$ which is also lower semicontinuous and superlinear. 

Since $\eta_{C}$ is linear by Proposition~\ref{kcisacone}, $x\in r \eta_{C}^{-1}(U)$ if and only if $\eta_{C}(x)\in rU$ for $x\in C$ and $r\in \mathbb R_{+}$. This implies that $M^{U} \circ\eta_{C}= M^{\eta_{C}^{-1}(U)}$, hence by uniqueness of extension, 
 $\overline{M^{\eta_{C}^{-1}(U)}}=M^{U}$.
 It is then straightforward to verify that $x\in  \overline{ M^{W} }^{-1}((1, +\infty])\subseteq U$. We conclude by noticing that $ \overline{ M^{W} }^{-1}((1, +\infty])$ is an open convex subset of $\kc(C)$. 

The case that $C$ is locally linear can be proved similarly by invoking Lemma~\ref{inlieu}, and the case that $C$ is weakly locally convex was proved in Proposition~\ref{kcisacone}. 
\end{proof}

\begin{remark}
In contrast with Remark~\ref{isobetweenxandkx}, we know from the proof of the above theorem that  for each convex open subset $W$ of a weakly locally convex topological cone~$C$, there exists one and only open convex open set~$V$ of $\kc(C)$ such that $\eta_{C}(W)= V\cap C$, and $\overline{M^{W}} = M^{V}$.
\end{remark}

\section{Topping the Cantor tree}
Focusing on the $\K$-category $\D$ of monotone convergence spaces and continuous functions, we note that the monad $\mathcal V_{\D}$ gives rise to a natural monad over $\mathbf{DCPO}$. 
\[
\begin{tikzcd}
\mathbf{DCPO}
\arrow[rr, "\Sigma"{name=F}, bend left=10] 
&&
\D \ar[rr, "F" {name=H}, bend left = 10] 
\arrow[loop, "\mathcal V_\D"', distance=2.5em, start anchor={[xshift=1ex]north}, end anchor={[xshift=-1ex]north}]{}{}
\arrow[ll, "\Omega"{name=G}, bend left=10]
\arrow[phantom, from=F, to=G, "\dashv" rotate=-90]
&& \D^{\mathcal V_{\D}} \ar[ll, "U" {name=K}, bend left = 10] 
\arrow[phantom, from=H, to=K, "\dashv" rotate=-90]
\end{tikzcd}
\]
The key to this observation is an adjunction $(\Sigma, \Omega)$ between $\mathbf{DCPO}$ and $\D$: The left adjoint $\Sigma$ sends each dcpo $L $ to $\Sigma L$, the space  $L$ equipped with the Scott topology, and the right adjoint $\Omega$ sends each monotone convergence space $X$ to the dcpo $\Omega X$. 
Since $\mathcal V_{\D}X$ is a monotone convergence space, we can restrict $\mathcal V_{\D}$ on the category $\D$, and it is again a monad over $\D$.
It is then immediate that the functor $\Omega \circ \mathcal V_{\D}\circ \Sigma$ is a monad over  {\bf DCPO}: Just note that $\mathcal V_{\D}$ can be written as $U \circ F$, where $F\dashv U$ is the canonical adjoint pair between $\D$ and the Eilenberg-Moore category $\D^{\mathcal V_{\D}}$ of $\mathcal V_{\D}$ that recovers the monad $\mathcal V_{\D}$. Then one easily verifies that $F\circ \Sigma\dashv \Omega\circ U$, and $\Omega \circ \mathcal V_{\D}\circ \Sigma =(\Omega\circ U)\circ (F\circ \Sigma)$.

Since we know that $\mathcal V_{\D} X$ is the $\D$-completion of $\vs X$, one may wonder whether $\Omega(\mathcal V_{\D} X) $ is the dcpo-completion of $\Omega(\vs X)$.   We see in the following example that the answer is no. The reason is that  the embedding of $\vs X$ into $\V X$ is not Scott-continuous in general, even when $X$ is an algebraic lattice endowed with its Scott topology.

\begin{example}
Consider the Cantor tree $\mathcal C$ consisting of finite and infinite words over $\{0, 1\}$ and endowed with the prefix order. $\mathcal C$ is a bounded complete algebraic domain with the empty word $\epsilon$ as the least element. The Cantor set can be realised as the set $C$ of maximal elements in  $\mathcal C$, that is, the set of all infinite words. By adding a top element~$\top$ to $\mathcal C$, the resulting poset, which we denote by $\mathcal C^\top$, is an algebraic lattice.
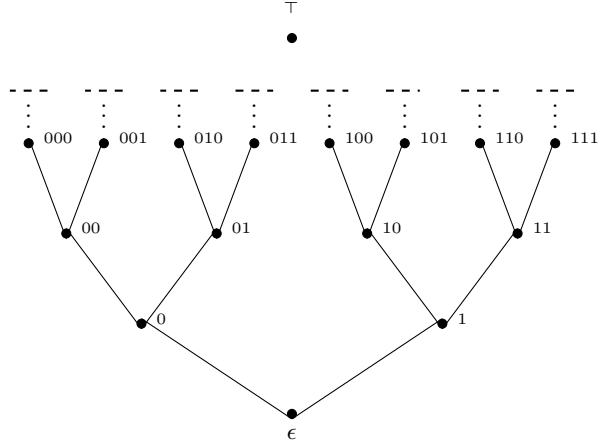
\begin{center}
\tikzset{
  solid node/.style={circle,draw,inner sep=1.2,fill=black},
  hollow node/.style={circle,draw,inner sep=1.2},
}
\begin{tikzpicture}[rotate=180,font=\footnotesize]
\draw[thick,dashed] (-3.75,-4.3) -- (-3.25,-4.3); 
\draw[thick,dashed] (-2.25,-4.3) -- (-2.75,-4.3); 
\draw[thick,dashed] (-1.25,-4.3) -- (-1.7,-4.3); 
\draw[thick,dashed] (-.25,-4.3) -- (-.75,-4.3); 
\draw[thick,dashed] (.75,-4.3) -- (.25,-4.3); 
\draw[thick,dashed] (1.25,-4.3) -- (1.75,-4.3); 
\draw[thick,dashed] (2.25,-4.3) -- (2.75,-4.3); 
\draw[thick,dashed] (3.25,-4.3) -- (3.75,-4.3); 
  \tikzset{
    level 1/.style={level distance=12mm,sibling distance=40mm},
    level 2/.style={level distance=12mm,sibling distance=20mm},
    level 3/.style={level distance=12mm,sibling distance=10mm},
    level 4/.style={level distance=12mm,sibling distance=7mm},
  }
  \node(l0)[solid node,label=below:{${\epsilon}$}]{}
    child{node[solid node,label=right:{$^{1}$}]{}
      child{node(l1)[solid node,label=right:{$^{11}$}]{}
        child{node(R5)[solid node,label=above:{$\vdots$},label= right:{$^{111}$}]{}
        		edge from parent node[left]{}}
        child{node[solid node,label=above:{$\vdots$},label= right:{$^{110}$}]{}
        		edge from parent node[left]{}}
        edge from parent node[left]{}
      }
      child{node(l2)[solid node,label=right:{$^{10}$}]{}
        child{node[solid node,label=above:{$\vdots$},label= right:{$^{101}$}]{}
        		edge from parent node[left]{}}
        child{node[solid node,label=above:{$\vdots$},label= right:{$^{100}$}]{}
        		edge from parent node[left]{}}
        edge from parent node[left]{}
      }
      edge from parent node[left]{}
    }
    child{node[solid node,label=right:{$^{0}$}]{}
      child{node(r1)[solid node,label=right:{$^{01}$}]{}
        child{node[solid node,label=above:{$\vdots$},label= right:{$^{011}$}]{}
        		edge from parent node[left]{}}
        child{node[solid node,label=above:{$\vdots$},label= right:{$^{010}$}]{}
        		edge from parent node[left]{}}
        edge from parent node[left]{}
      }
      child{node(r2)[solid node,label=right:{$^{00}$}]{}
        child{node[solid node,label=above:{$\vdots$},label= right:{$^{001}$}]{}
        		edge from parent node[left]{}}
        child{node(L5)[solid node,label=above:{$\vdots$},label= right:{$^{000}$}]{}
        		edge from parent node[left]{}}
	edge from parent node[left]{}
      }
      edge from parent node[left]{}
    }
  ;
  \draw  (0,-5) node[solid node,label=above:{$^\top $}]{};
\end{tikzpicture}
\end{center}

 For each natural number $n$, we define the \emph{normalized counting measure at level $n$} as: $$c_{n}= {1\over 2^{n}} \Sigma_{x\in F_{n}} \delta_{x},$$ where $F_{n}$ is the set of words in $\mathcal C$ with length $2^n$. It is obvious that in $\V \mathcal C^{\top}$ the supremum $\sup_{n} c_{n}$ is not a simple valuation. In fact, it is well known that $\sup_{n}c_{n}$ is equal to $\V(e)(\mu)$, where $e$ is the topological embedding of $C$ into $\mathcal C^{\top}$ (with the Scott topology) and $\mu$ is the Haar measure on $C$, regarded as an infinite product of two-point groups (see~\cite{mislove20}). 
 
We claim that in $\vs \mathcal C^{\top}$, the supremum of $c_{n}, n=1, 2, 3,...$ is~$\delta_\top$:

First, $\delta_\top$ is obviously an upper bound of $c_{n}$. 

Since $\da \delta_\top$ is Scott closed in $\vs \mathcal C^\top$, for any subset of $\da \delta_\top$, there is no difference between computing its supremum in $\da \delta_\top$ and in $\vs \mathcal C^\top$.

Let $\Sigma_{x_{i}\in F} r_i \delta_{x_i}$, $F$ finite, be a simple valuation. We show that if $x_i \not= \top$ for some index $i$, or if $\sum_i r_i < 1$, then $\Sigma_{x_{i}\in F} r_i \delta_{x_i}$ is not an upper bound for some $c_n$. Indeed, if $\sum_i r_i < 1$, then $\Sigma_{x_{i}\in F} r_i \delta_{x_i}({\mathcal C}^\top) < 1 = c_n({\mathcal C}^\top)$, so $\Sigma_{x_{i}\in F} r_i \delta_{x_i}$ is not above any of the $c_n$s. 

On the other hand, assume $\sum_i r_i =1$, but $x_i < \top$ for some $i$. Then $U = {\mathcal C}^\top \setminus \da x_i$ is an open set, and $\Sigma_{x_{i}\in F} r_i \delta_{x_i}(U) \leq 1 - r_i$. Then we can choose an $n$ large enough that $c_n(U) > 1 - r_i$: 

Indeed, if $x_i\not\in C$, then $x_i \in F_m$ for some $m$, and choosing any $n > m$ implies $c_n(U) = 1$. Hence $c_n \not\leq \Sigma_{x_{i}\in F} r_i \delta_{x_i}$.

If $x_i\in C$, then we can choose an $n > 0$ large enough such that $b$ out of $2^{n}$ elements of $F_{n}$ are not below $x_{i}$ and $b/2^{n}$ is strictly greater than $1-r_{i}$. Hence $ \Sigma_{x_{i}\in F} r_i \delta_{x_i}(U) \leq 1-r_{i}<b/2^{n}\leq c_{n}(U)$, this again shows that $c_n\not\leq \Sigma_{x_{i}\in F} r_i \delta_{x_i}$.

As a result, for $\Sigma_{x_{i}\in F} r_i \delta_{x_i}$ to be an upper bound for the $c_n$s, we must have $x_i = \top$ for all $i$, and $\sum_i r_i = 1$; i.e., $\Sigma_{x_{i}\in F} r_i \delta_{x_i} = \delta_\top$.
\end{example}

\begin{remark}
The aforementioned example also shows that we cannot define $\vs$ as a functor from the category of dcpos and Scott-continuous functions to that of posets and Scott-continuous functions. 

Consider the Scott-continuous map $f\colon \mathcal C^\top \to \mathbb S$ that sends $\mathcal C$ to $0$ and $\top$ to $1$, where $\mathbb S$ is the Sierpi\'nski space. Then $\vs (f)\colon \vs(\mathcal C^\top)\to \vs (\mathcal S)$, defined as $\vs(f)(s)(U)=s(f^{-1}U)$ for $s\in \vs(\mathcal C^\top)$ and $U$ open in $\mathbb S$, is not Scott-continuous. As discussed above we know that  $\sup_{n\in \mathbb N} c_n = \delta_\top$ in $\vs \mathcal C^{\top}$. However, with straightforward computation we have that $\sup_{n\in \mathbb N} \vs(f)(c_n)(\oneset{1}) = 0$ and $\vs(f)(\delta_\top)(\oneset{1})= 1$.
\end{remark}

\section*{Acknowledgement} The authors thank Andre Kornell and Bert Lindenhovius for several helpful conversations. We also gratefully acknowledge the support of the US AFOSR under MURI award FA9550-16-1-0082.

\newcommand{\etalchar}[1]{$^{#1}$}

\end{document}